\pdfoutput=1
\documentclass[sigconf, authorversion]{acmart}

\copyrightyear{2018}
\acmYear{2018}
\setcopyright{acmlicensed}
\acmConference[WPES'18]{2018 Workshop on Privacy in the Electronic Society}{October 15, 2018}{Toronto, ON, Canada}
\acmBooktitle{2018 Workshop on Privacy in the Electronic Society (WPES'18), October 15, 2018, Toronto, ON, Canada}
\acmPrice{15.00}
\acmDOI{10.1145/3267323.3268947}
\acmISBN{978-1-4503-5989-4/18/10}

\renewcommand\footnotetextcopyrightpermission[1]{} 
\keywords{E-mail encryption; Decentralization; Key distribution; Privacy}

\usepackage{common}

\title{ClaimChain: Improving the Security and Privacy of\\ In-band Key Distribution for Messaging}

\author{Bogdan Kulynych}
\affiliation{EPFL SPRING Lab}
\email{bogdan.kulynych@epfl.ch}

\author{Wouter Lueks}
\affiliation{EPFL SPRING Lab}
\email{wouter.lueks@epfl.ch}

\author{Marios Isaakidis}
\affiliation{University College London}
\email{m.isaakidis@cs.ucl.ac.uk}

\author{George Danezis}
\affiliation{University College London}
\email{g.danezis@ucl.ac.uk}

\author{Carmela Troncoso}
\affiliation{EPFL SPRING Lab}
\email{carmela.troncoso@epfl.ch}

\begin{abstract}
The social demand for email end-to-end encryption is barely supported by mainstream service providers. Autocrypt is a new commu\-ni\-ty-driven open specification for e-mail encryption that attempts to respond to this demand. In Autocrypt the encryption keys are attached directly to messages, and thus the encryption can be implemented by email clients without any collaboration of the providers. The decentralized nature of this in-band key distribution, however, makes it prone to man-in-the-middle attacks and can leak the social graph of users. 
To address this problem we introduce ClaimChain, a cryptographic construction for privacy-preserving authentication of public keys. Users store claims about their identities and keys, as well as their beliefs about others, in ClaimChains. These chains form authenticated decentralized repositories that enable users to prove the authenticity of both their keys and the keys of their contacts. ClaimChains are encrypted, and therefore protect the stored information, such as keys and contact identities, from prying eyes. At the same time, ClaimChain implements mechanisms to provide strong non-equivocation properties, discouraging malicious actors from distributing conflicting or inauthentic claims. We implemented ClaimChain and we show that it offers reasonable performance, low overhead, and authenticity guarantees.

\end{abstract}

\begin{document}

    \maketitle

    \section{Introduction}
\label{section:intro}

Following the Snowden revelations it became clear that, given the dependence of citizens, governments, and corporations on electronic communications, there is a strong need for highly secure end-to-end encrypted communications. That is, the content of communications must not be accessed by third parties, so as to shield them from mass surveillance systems, domestic or foreign. Yet, so far we have only seen feeble and largely unsuccessful attempts by mainstream service providers such as GMail and Yahoo to support fully encrypted e-mail~\cite{googleabandon,googlevapor,yahoogoogle}. Only a few minor e-mail providers embrace end-to-end encryption.\footnote{\url{https://mailfence.com}}\footnote{\url{https://protonmail.com}}

To fill this void, a recently launched community-driven initiative, Autocrypt, is developing a new open specification for e-mail encryption. The goal is to facilitate key handling by mail user agents so that encryption can be deployed without the need for collaboration of the service providers. 
The Autocrypt approach is supported\footnote{\url{https://github.com/autocrypt/autocrypt/blob/master/doc/install.rst}} by key e-mail clients such as Thunderbird+Enigmail, K-9 mail, and Mailpile, as well as a new messaging application for Android, DeltaChat.

Similarly to in-band PGP~\cite{Zimmerman95}, Autocrypt embeds the encryption keys into the e-mail messages, but uses pre-defined headers instead of attachments. Furthermore, in the spirit of the PGP Web of Trust, these headers also contain cross-references to the keys of other users. The cross-references implicitly endorse the binding between these keys and the corresponding user identities.

This decentralized approach alleviates the privacy problem of centralized certification authorities, such as SKS Keyservers\footnote{\url{https://sks-keyservers.net}} for PGP keys, which can observe users' key look-ups, and thus can infer their communication patterns. However, since no one has a global view of all the bindings in Autocrypt's decentralized approach, malicious users or providers can supply different user-to-key bindings to different recipients, effectively opening the doors to man-in-the-middle attacks.

This attack, whereby Alice can show to Carol and Donald different versions of Bob's key to manipulate encryption in her advantage, is commonly known as \emph{equivocation}.  
A solution to render equivocation detectable and accountable could be to use CONIKS~\cite{MelaraBBFF15}. However, CONIKS' transparency logs are maintained by providers, and thus the scheme is not compatible with the Autocrypt principle of not requiring provider collaboration.

In this paper we present \emph{\Keychain}, a cryptographic construction that alleviates the authenticity and privacy problems of in-band key distribution in the setting of Autocrypt.\footnote{\Keychains are currently being tested by the Autocrypt team (\url{https://py-autocrypt.readthedocs.io/en/latest/})} Similarly to CONIKS, \keychains consist of chained blocks. Instead of a global log, however, each user has their own \keychain that contains all the information necessary to represent her claims about her own keys, and her beliefs about other users' keys, i.e., her cross-references. Chaining of blocks enables tracking and authenticating the evolution of beliefs and keys. Cross-references enable users to combine their contacts' beliefs---represented by their \keychains---to establish evidence about the binding between identities and keys. 

To address the privacy issues of the Web of Trust cross-reference sharing model, \keychains use cryptographic access tokens to provide fine-grained control on who is allowed to read which claims. Moreover, \keychains' claim encoding schemes make it hard to infer how users' beliefs change over time. Finally, \keychains include mechanisms to reliably and efficiently prevent equivocation within a block, and a mechanism to detect equivocation across blocks. In doing so, \keychains ensure non-equivocation while minimizing the leakage of users' friendship networks.

We prove the security and privacy properties provided by \keychains. We also provide an implementation of \keychains
and show that it scales to accommodate the needs of large groups at an acceptable overhead cost. We also simulate the usage of \keychains for in-band key distribution using the Enron e-mail dataset\footnote{\url{http://www.cs.cmu.edu/~enron/}} in a privacy-preserving way: \Keychain owners reveal cross-references only if that does not leak more information about their social graph than is already revealed by the e-mails themselves. We quantify the degree to which in-band key distribution can protect communication. We show that \keychains improve privacy, without diminishing considerably the ability to encrypt e-mail, and enable the detection of incorrect key information.
Our main contributions are:
\begin{itemize}[noitemsep,nolistsep,leftmargin=\parindent,label=--]
    \item We introduce \keychain, a cryptographic construction based on authenticated data structures. \keychains store claims about keys, thereby supporting key authentication in decentralized environments in a secure and privacy-preserving way.
    \item We define the properties that a decentralized and privacy-pre\-ser\-ving key distribution system should offer and we formally model them for \keychain. 
    \item We show that owners cannot equivocate about contact keys within blocks. Moreover, we provide a novel mechanism that enables \keychain owners to prove that they have not equivocated across different blocks about a particular contact key. Unlike other transparency-backed solutions~\cite{MelaraBBFF15}, auditing the consistency of contact keys is possible without revealing their actual values in the \keychain history.
    \item We provide an implementation of \keychain and show that its computation and bandwidth requirements are reasonable and within reach of modern computers and networks.
    \item We evaluate the effectiveness of decentralized key distribution on a real e-mail dataset. We show that selective privacy-preserving distribution can be almost as effective as broadcasting all known keys, although total encryption is difficult to achieve.
\end{itemize}

    \section{Problem statement and goals}
\label{sec:problem_statement}

We assume a messaging system in which users embed their cryptographic keys in-band, i.e., into the messages themselves or into the message headers, as in Autocrypt. These keys are used to provide message confidentiality using opportunistic encryption~\cite{rfc7435}. That is, the communication is encrypted when users know each others' keys, but falls back to plaintext when they do not. 

Sending keys as part of message headers results in two problems. First, in terms of privacy, adding such headers reveals users' social ties. Second, in terms of security, man-in-the-middle attackers can modify the header contents, since they are not authenticated. Moreover, malicious users can equivocate about others' keys.

\descr{Design goals.} 
We assume that all actors in the system, users and providers, may act maliciously.  
Our goal is to design a data structure that can store the binding between keys and identities, and is suitable for integrating with in-band key distribution. The purpose of the structure is to support key validation, i.e., help users establish the authenticity of user-key bindings, as long as some users in the system are honest. Furthermore, it must protect users' privacy without relying on centralized parties. 

More concretely, we aim at providing the following properties. First, the structure must guarantee the \emph{integrity} and \emph{authenticity} of identity-key bindings, i.e., it should not be possible to replace or inject bindings without being detected.
Second, we want to preserve the \textit{privacy of cross-referenced information} and the \textit{privacy of the social graph}. These properties ensure that only authorized users can access the key material in the structure and the identities of the bindings being distributed. 
Third, the structure must prevent users from \textit{equivocating} other users with respect to the identity-key bindings that they share. That is, a user Owen should \emph{not} be able to show to Alice and Bob different versions of a Charlie's key, even if he withdraws Alice's access to see Charlie's keys. In the latter case if Alice ever regains access, she must be able to detect Owen's misbehavior.
Finally, our construction should not entail significant computational or communication overhead for the end users and providers to enable adoption at large scale.

\descr{Non-goals.} 
In-band key distribution cannot ensure full availability of public encryption keys. The keys of one or more recipients may not be available to a sender at a time of sending a message, and thus, because of the opportunistic encryption operation, the message would be sent in the clear.
We are therefore not concerned with ensuring 100\% availability of keys. Instead, our goal is to secure the keys that \textit{are} distributed without harming privacy. If guaranteeing encrypted communication is absolutely necessary, parties must exchange keys in a reliable way, e.g. through a centralized service or an out-of-band mechanism. 

Furthermore, throughout this paper we consider that users have only one identity, and use one and only one structure to store key bindings of their contacts. If a user wishes to have different identities, she must create one structure per identity.
    
\section{\Keychain design}
\label{sec:claimchain_design}

In this section we introduce \keychain, a structure to store key bindings in a secure and privacy-friendly manner. 

\subsection{Cryptographic preliminaries and notation}\label{sec:prelims}

We denote sampling uniformly at random from a set $X$ as $x \sample X$, and the assignment of an evaluation of a function $f(x)$ to $y$ as $y \gets f(x)$, regardless of whether $f$ is probabilistic or deterministic. We denote concatenation of strings by $\parallel$.

Let $\secparameter$ be the security parameter. \Keychain relies on the following standard cryptographic primitives. 
Let $\enc(k, m) \mapsto c$ and $\dec(k, c)$
denote an IND-CPA secure symmetric authenticated encryption scheme.
\Keychain uses an existentially unforgeable signature given by the algorithms 
$\sig.\keygen(1^\lambda)$ returning the keypair $(\sksig, \pksig)$, $\sgn(\sksig, m)$ returning a signature $\sigma$, and the verification function $\sig.\vrfy(\pksig, \sigma, m) \mapsto \{\top, \bot\}.$ We write $\mydh.\keygen(1^\lambda)$ for the generation of a Diffie-Hellman (DH) keypair $(\skdh, \pkdh)$ using which we can non-interatively compute the shared DH key $s \in \strs$ using $\derivesecret(\skdh, \pkdh^R)$. Finally, let $H$ be a cryptographic hash function from which we derive a family of hash functions~$H_i: \strs \to \{0,1\}^{2\secparameter}, i > 0$.


All schemes use a cyclic group $\group$ of prime order $\grouporder$ generated by $g$. We write $\Zq$ for the integers modulo $\grouporder$. Moreover, we assume the existence of a cryptographic hash function $H_{\group}: \strs \to \group$ that hashes strings to group elements, and a hash function $H_q: \strs \to \Zq$ that hashes strings to the elements of $\Zq$.

\Keychains also require an information-theoretically hiding commitment scheme $\commit(r, m)$ that commits to values $m \in \Zq$ given a randomizer $r \in \Zq$. We instantiate this scheme using Pedersen's commitment scheme~\cite{Pedersen91}. Let $g_1, g_2 \in \group$ be random generators such that the discrete logarithms of $g_1$ and $g_2$ with respect to each other are unknown. Then, $\commit(r, m) = g_1^{r} g_2^{m}$.

\Keychains use standard zero-knowledge proofs of knowledge, and in particular Schnorr's proof of knowledge of discrete logarithms~\cite{Schnorr91}, to prove correctness of claims. We use the Fiat-Shamir heuristic~\cite{FiatS86} to derive non-interactive signature proofs of knowledge. For example, we write:
\begin{equation*}
    \textrm{\textsf{SPK}}\left\{ (r, m): C = g_1^{r} g_2^{m} \right\}(t)
\end{equation*}
to denote the non-interactive signature proof of knowledge on a random string $t$ for which the prover knows the commitment opening $(r, m)$. To focus on the semantics of the proof, we write
\begin{equation*}
    \textrm{\textsf{SPK}}\left\{ (r, m): C = \commit(r, m) \right\}(t)
\end{equation*}
instead, to denote the same proof.

Finally, \keychains use a verifiable random function (VRF)~\cite{MicaliRV99, FranklinZ13}, given by the algorithms $\vrf.\keygen$ and $\vrf.\eval$. The function $\vrf.\keygen(1^{\secparameter})$ returns a keypair $(\skvrf,\allowbreak \pkvrf) = (\skvrf,\allowbreak g^{\skvrf})$. Then, $h = \vrf.\eval(\skvrf,\allowbreak m) = H_\group(m)^{\skvrf}$ is the VRF of the value $m$. Users prove the that $h$ was correctly computed by constructing the proof
\begin{equation*}
    \textsf{SPK}\left\{ (\skvrf): \pkvrf = g^{\skvrf} \land h = \vrf.\eval(\skvrf, m) \right\}().
\end{equation*}
The properties of VRF hashes are similar to that of cryptographic hashes: uniqueness of $h$ for a given message and private key, collision resistance, and pseudorandomness (assuming no access to the corresponding proof)~\cite{PapadopoulosWHN17}.

\begin{figure}
    \centering
    \begin{tikzpicture}[%
    block/.style={
      draw, rectangle, ultra thin,
      minimum height=1.25em,
      text centered,
      text width=2.2em,
      text depth=0.5em,
      fill=white,
      text height=1em,
    },
    payload/.style={
      draw, rectangle, ultra thin,
      minimum height=2em,
      align=left,
      text width=8.5em,
      inner sep = 3pt, 
    },
    chainlabel/.style={align=center,font=\small\itshape},
    box/.style={thick}]

    \matrix (bi) [matrix of math nodes,column sep=-0.125pt,row sep=5mm,nodes={block}]{
      X_{i+1} & \ptr_{i+1} & \sigma_{i+1} \\
      X_{i} & \ptr_{i} & \sigma_{i} \\[5mm]
      X_{0} & - & \sigma_{0} \\
    };
    \node[right=2mm of bi-1-3] {$B_{i+1}$};
    \node[right=2mm of bi-2-3] {$B_{i}$};
    \node[right=2mm of bi-3-3] {$B_0$};
    \node[above=-1mm of bi,text depth=0.2em] {Chain};

    \matrix (block) [left=15mm of bi.north west,
        yshift=0em,
        anchor=north east,
        matrix of nodes,
        row sep=-0.125pt,
        nodes={payload},
        row 3 column 1/.style={nodes={payload,minimum height=3.5em}}]{
      Block index: $i$ \\
      Nonce: $\nonce$ \\
      {Metadata:\\ $\pksig, \pkvrf, \pkdh$} \\
      Public data \\
      Block map \\
    };
    \node[above=-1mm of block,text depth=0.2em] {Payload $X_i$};

    \node[below=0mm of bi-2-2] {$\vdots$};

    \draw[->] (bi-1-2) -- (bi-2-2);

    \draw[box] (bi-1-1.north west) rectangle (bi-1-3.south east);
    \draw[box] (bi-2-1.north west) rectangle (bi-2-3.south east);
    \draw[box] (bi-3-1.north west) rectangle (bi-3-3.south east);
    \draw[box] (block-1-1.north west) rectangle (block-5-1.south east);

    \begin{scope}[on background layer]
        
        \fill[gray] (bi-2-1.north east) -- (block-1-1.north east) -- (block-5-1.south east) -- (bi-2-1.south east) -- cycle;
    \end{scope}

    \end{tikzpicture}
    \caption{\Keychain block structure}
    \label{fig:block_structure}
\end{figure}

\subsection{Overview}

We consider that each user has a state made of information about herself and her beliefs about other users' states. 
At a given point in time a users' state is represented as a set of statements, called \emph{claims}. Claims can be of two kinds. The first type of claim refers to a user's own state. In particular, these may be statements on the user's encryption keys, identity information (screen name, real name, or e-mail), or other cryptographic material such as verification keys to support digital signatures.
The second type of claims, we call them cross-references, refer to other users' states.
A claim owner creates a cross-reference to endorse the referenced user's state as being authoritative, i.e., a cross-reference indicates the owner's belief that the self key material found in those users' state is correct.
A user's state evolves over time as she rotates her keys and observes the evolution of others' states. She stores snapshots of her state in a cryptographic data structure called a \emph{\keychain}.

The core element of a \keychain is a block. A block includes all claims that the owner endorses at the time when she creates the block, i.e., a block is a snapshot of the owner's state. Blocks form a chain. A block contains a payload $X$, a pointer to the previous block $ptr$, and a digital signature $\sigma_i$ on the payload and the pointer. See Figure~\ref{fig:block_structure}. 
The payload of the previous block $X_{i-1}$ contains the verification key $\pksig^{(i-1)}$ for the private key $\sksig^{(i-1)}$ that signs $\sigma_i$.

We now describe each of the block components in detail.

\descr{The payload $X_i$} has the following content (see Figure~\ref{fig:block_structure}, left):

\begin{itemize}[noitemsep,nolistsep,leftmargin=\parindent,label=--]

\item \textit{Block index.} The block's position in the chain. The index of the genesis block is~0.

\item \textit{Nonce.} A fresh cryptographic nonce used to `salt' all cryptographic operations within the block. It ensures that the information across blocks is not linkable. 

\item \textit{Metadata.} The current signature verification key of the owner $\pksig$, that is used to authenticate the next block of the \keychain; the current key $\pkvrf$ to compute a verifiable random function used to support non-equivocation; and a Diffie-Hellman key $\pkdh$ used to provide claim privacy. 

\item \textit{Public data.} Application-specific data the owner wishes to make publicly visible. For in-band key distribution we set this to the owner's self-claim on her current public encryption key.

\item \textit{Block map.} A high-integrity key-value map storing the claims, as well as access tokens that express access-control rights. This map 
has two core properties: i) a key can only be resolved to a single value, and ii) it enables the generation and verification of efficient proofs of inclusion of claims or access tokens. We implement the map using \emph{unique-resolution key-value Merkle trees}, explained in more detail in Section~\ref{sec:block_map}. For our use case the map only contains cross-references.


\end{itemize}


\descr{The signature }$\sigma_i = \sgn(\sksig^{(i-1)}, (X_i, \ptr_i))$ authenticates the current block. 
A block $B_i$ must have a valid signature under the verification key indicated in the payload of the previous block $B_{i-1}$. The genesis block of a \keychain is `self-signed'. The corresponding initial public signing key is included in the initial payload. 
Each block in the chain contains enough information to authenticate past blocks as being part of the chain, validate the next block, and, by transitivity, all future blocks as being valid updates. Therefore, a user with access to a block of a chain that she believes is authoritative, 
can both audit past states of the chain, and authenticate the validity of newer blocks. 



\subsection{Low-level operations}\label{sec:details}

We now describe how we implement claims and access tokens, and how they are combined into the block map. 

\descr{Claims.}
\label{sec:claimdefinition}
%
%
We model claims as a tuple composed of a \emph{label} $l$ and a \emph{body} $m$. The label is a well-known identifier associated with the identity of the user to whom the claim refers. The body is the state of that user at the time when the claim is generated, represented as the latest block of this user's \keychain. For instance, a claim (`bob@gmail.com', $B$) represents the \keychain owner's belief that the current state of the user associated with this Gmail account is represented by the block $B$.

\begin{figure*}[t]
\fbox{
\begin{minipage}[t]{.48\linewidth}
\begin{algorithmic}[1]
\Procedure{EncClaim}{$\skvrf, l, m, \nonce$}
    \Let{$h$}{$\vrf.\eval(\vrftokenExpr)$}
    \Let{$i$}{$H_1(h)$}
    \State{$r \sample \Zq,\quad \claimcommit \gets \commit(r, H_q(m))$}
    \Sample{$\proofkey$}{$\{0,1\}^{\secparameter}$}
    \Let{$\pi$}{$\textrm{\textsf{SPK}}\{ (\skvrf, r): \pkvrf = g^{\skvrf} \land
    \newline\mbox{}\qquad\qquad
    h = \vrf.\eval(\vrftokenExpr) \land
    \newline\mbox{}\qquad\qquad
    \claimcommit = \commit(r, H_q(m))\}(\proofkey)$}
    \State{$\claimkey \sample \{0,1\}^{\secparameter},\quad \claim \gets \enc(\claimkey, \pi \parallel m) \parallel \claimcommit$}
    \State \Return $r, h, \claimkey, \proofkey, (i, c)$
\EndProcedure
\label{alg:lowlevel}
\end{algorithmic}

\begin{algorithmic}[1]
\Procedure{EncCap}{$\skdh, \pkdh^R, l, h, \claimkey, \proofkey, \nonce$}
    \Let{$s$}{$\derivesecret(\skdh, \pkdh^R)$}
    \Let{$i_\capab$}{$H_3(s \parallel l \parallel \nonce)$}
    \Let{$k_\capab$}{$H_4(s \parallel l \parallel \nonce)$}
    \Let{$\capab$}{$\enc(k_\capab, h \parallel \claimkey \parallel \proofkey)$}
    \State \Return $(i_\capab, \capab)$
\EndProcedure
\end{algorithmic}

\end{minipage}%
\begin{minipage}[t]{.48\linewidth}
\begin{algorithmic}[1]
\Procedure{DecClaim}{$\pkvrf^O, h, l, \claimkey, \proofkey, c, \nonce$}
    \Let{$\bar{c} \parallel \claimcommit$}{$c$}
    \Let{$\pi \parallel m$}{$\dec(\claimkey, \bar c)$}
    \LineComment{{\footnotesize Note the verification of $\pi$ requires $\pkvrf^O$, $h$, $l$, $\proofkey$, $\claimcommit$, $m$, and $\nonce$.}}
    \If {$\pi$ is not a valid proof}
        \State \Return $\bot$
    \EndIf
    \State \Return $m$
\EndProcedure
\end{algorithmic}
\mbox{}\\[4.7mm]
\begin{algorithmic}[1]
\Procedure{DecCap}{$\skdh, \pkdh^O, l, \capab, \nonce$}
    \Let{$s$}{$\textsf{SharedSecret}(\skdh, \pkdh^O)$}
    \Let{$k_\capab$}{$H_4(s \parallel l \parallel \nonce)$}
    \Let{$h \parallel \claimkey \parallel \proofkey$}{$\dec(k_\capab, \capab)$}
    \Let{$i$}{$H_1(h)$}
    \State \Return $i, h, \claimkey, \proofkey$
\EndProcedure
\end{algorithmic}

\end{minipage}
}
\caption{Low-level \keychain operations}
\label{fig:low_level_algorithms}

\end{figure*}

For privacy reasons, claims in a \keychain are encrypted. Thus, they cannot be found directly by other users. To enable efficient search for concrete claims within a \keychain block, we introduce a \emph{lookup key}, or \emph{index}, $i$ for each claim.

We illustrate the encoding of claims in procedure \Call{EncClaim}{}, see Figure~\ref{fig:low_level_algorithms}. Consider a claim ($l$, $m$) that is to be included in a block. We first compute the unique VRF of its label and derive the claim's lookup key $i$ [lines 2--3]. Note that the computation of $h$ includes a per-block nonce to ensure that the lookup keys for a given claim label look different across blocks, and therefore no patterns can be inferred from their appearance. 

Recall that VRF hashes are unique. We use them to derive lookup keys to ensure that, given a label, all users retrieve the same claim, effectively supporting non-equivocation within a block. 
This use of VRFs is inspired by CONIKS~\cite{MelaraBBFF15}. We also include additional cryptographic elements in our encoded claims in order to obtain a stronger non-equivocation property than CONIKS. Specifically, we guarantee that equivocation is detectable \emph{across blocks} without the need for key owners to intervene. 
The need for a detection mechanism stems from the fact that \keychain owners can give and withdraw access to claims at will. Thus, they can try to equivocate others by giving them access to different information in different blocks. To make this misbehaviour detectable we provide \keychain owners with the ability to prove statements about claims that other users cannot see. This way, if a proof cannot be completed, equivocation is revealed (see Section~\ref{sec:inter_block_equivocation} for more details).

To prove statements on claim contents without revealing them, 
we commit to the claim body $m$ [line~4]. Moreover, we construct a non-interactive proof $\pi$ on $\proofkey$ proving that the VRF $h$ is correct and that the commitment $\claimcommit$ commits to $m$ [lines~5--6]. When decoding a claim, users verify the proof $\pi$. The proof verification key $\proofkey$ ensures that only authorized users can verify this proof.

Once $\pi$ is computed, we encrypt $m$ and $\pi$ with a random key~$\claimkey$ [line~7]. Finally, the claim encoding consists of this ciphertext and the commitment \claimcommit: $c = \enc(\claimkey, \pi \parallel m) \parallel \claimcommit.$

The binding property of the commitment $\claimcommit$ and the validation provided by the proof $\pi$ also ensure that all users with access to an encoded claim $c$ must recover the same claim body~$m$.
This makes this encoding scheme an instance of committing, or non-deniable, encryption~\cite{GrubbsLR17}. Hence, a malicious owner can not equivocate by supplying two different claim encryption keys to different users.

The procedure \Call{DecClaim}{}, see Figure~\ref{fig:low_level_algorithms}, describes the decoding of a claim. It takes as input the encryption key $k$, the VRF hash $h$, and the proof verification key $\proofkey$ from the owner (see below). Then, users can decrypt the ciphertext using $k$ [lines 2--3]; and verify the claim proof $\pi$, which includes the verification of the correctness of the VRF hash $h$ and of the commitment~$\claimcommit$ [lines 4--6]. 


Our claim encoding scheme offers four distinct security advantages. First, the use of the VRF ensures that lookup keys can only be produced by the owner of the chain, which as we describe below supports access control. Second, the lookup key is unique for a given label, and thus can be used to support non-equivocation for claims within a block. Third, the lookup key $i$ and claim encoding $c$ leak no information about the claim label or body. Fourth, it supports zero-knowledge proofs about claim contents, which enables the detection of equivocation across blocks.

\descr{Access capabilities.}\label{sec:access_control}
\Keychain owners create cryptographic access tokens called \emph{capabilities} to ensure that only authorized users can access specific claims. A single capability grants one authorized user access to one claim. We call the authorized users \emph{readers}.

An encoded capability is an encryption of all the values needed to obtain a claim lookup key and decode the corresponding claim: the encryption key $k$, the VRF hash $h$, and the proof verification key $\proofkey$. We encrypt these using a key derived from a shared secret~$s$ between the chain owner and the reader. Similarly to claims, encoded capabilities have an associated lookup key $i_\capab$, and a body~$\capab$. 

The procedure \Call{EncCap}{}, see Figure~\ref{fig:low_level_algorithms}, describes how to encode capabilities. First, it computes the shared Diffie-Hellman secret~$s$ using the owner's private DH key $\skdh$ and reader's public DH key $\pkdh^{R}$ [line 2]. The latter is available in the metadata of the reader's \Keychain. We use the secret $s$ to derive both the capability lookup key $i_\capab$ [line 3], and the capability encryption key $k_\capab$ [line 4]. Then we encrypt the values $h$, $k$, and $\proofkey$ using the key $k_\capab$ to obtain the capability encoding [line 5]: $\capab = \enc(k_\capab, h \parallel k \parallel \proofkey).$
 

Chain owners store the encoded claim $c$ under the lookup key $i$ in the block map. Similarly, they store the encoded capability $\capab$ under the lookup key $i_{\capab}$. To find a capability corresponding to a claim with label $l$ in a \keychain block, a reader first computes the lookup key $i_\capab$ for label $l$ using the shared secret with the \keychain owner. If the corresponding capability~$\capab$ is in the block, she decodes it using \Call{DecCap}{}, see  Figure~\ref{fig:low_level_algorithms}. First, the reader derives the shared secret~$s$ [line 2], and computes the capability encryption key $k_\capab$ using the claim label $l$ [line 3]. She can then decrypt $\capab$ using $k_\capab$ [line 4], obtaining the label's VRF hash $h$, the encryption key $k$, and the proof verification key $\proofkey$. With this information the reader can compute the claim's lookup key $i = H_1(h)$, find the claim, and decode it using \Call{DecClaim}{}.



\descr{Block map.}\label{sec:block_map} Encoded claims and capabilities are stored in the block map. 
We implement the block map using a unique-resolution key-value Merkle tree. Unlike a standard Merkle tree that implements an authenticated set data structure, a key-value tree is an instance of an authenticated dictionary~\cite{CrosbyW11}. It can be efficiently queried for a value that corresponds to a given lookup key. Our construction is similar to that of a binary search tree: the intermediate nodes contain pivots that define whether the querier should follow the left child or a right child; the leaf nodes contain the values. The construction allows queriers to be sure that retrieved values are unique, i.e., there cannot exist any other leaf nodes that correspond to the queried lookup key. We call this the \emph{unique resolution property.} We formally define the property in Experiment~\ref{exp:unique_resolution} and prove it in Theorem~\ref{thm:unique_resolution} (both in Appendix~\ref{app:structures}). We refer to Appendix~\ref{app:structures} for further details on the construction.

The unique-resolution property guarantees that for a given lookup key $i$, respectively $i_\capab$, there can only be one claim $c$, respectively capability $\capab$. The uniqueness of the VRF value $h$, the property of the tree, and the commitment in the claim encoding, ensures that a \keychain owner can not equivocate within a block.

We note that \Keychain blocks only need to include the root hash of the Merkle tree, not the whole tree. This is because our Merkle tree construction allows to produce an inclusion proof for items: a path from the root to the leaf node which contains the item. Thus, providing others with this paths is enough to convince them that the items are in the tree defined by the root in the block.

\subsection{High-level operations}
So far we have described how users can encode and decode claims. We now outline how these claims can be included in a \Keychain and read from it. At a glance, owners create blocks with a set of encoded claims and corresponding encoded capabilities, and use them to extend their \keychains. Any user can validate the authenticity and integrity of the chain. Moreover, readers can retrieve the claims they are authorized to read. Section~\ref{sec:using} illustrates how these operations can be used in the context of in-band key distribution.

\descr{Content-addressable store.} \Keychain owners store their blocks and trees in mutable \emph{content-addressable stores}. These are key-value stores where the key must be the hash of the corresponding value. They are a good fit for \keychains because i) it is easy to verify their integrity by checking that all keys are the hashes of the respective objects they map to; and ii) an incomplete store cannot lead to an erroneous decision on the authenticity, inclusion or exclusion of any block or tree node.
%
The store supports two operations:
\begin{itemize}[noitemsep,nolistsep,leftmargin=\parindent,label=--]
    \item \textsc{Put}($v$). Record the value $v$ in the store.
    \item \textsc{Get}($h$). Return $v$ such that $h = H(v)$, if present in the store.
\end{itemize}


\descr{Extending a chain.} Whenever an owner decides to add new claims to her \keychain she uses the procedure \Call{ExtendChain}{} in Figure~\ref{fig:extend_and_query}. This procedure takes as input the public application data, a set of claims $(l_j, m_j)$ to add to the block, an access control set $\acm$ consisting of the authorized reader-label pairs for these claims, the cryptographic keys necessary to create the block (keypairs for signatures, DH key exchange, and VRF), as well as the previous signing key $\sksig'$ included in the previous block, the pointer $\ptr$ to that block, and, finally, the user's store.


To create a block the user first generates a random nonce that is used for all encoding operations [step~1]. She then encodes all the claims and capabilities [steps~2--3]. The set $S$ of encoded values and their respective lookup keys are used to construct a Merkle tree with root hash \textsf{MTR}, as described in Algorithm~\ref{alg:build_tree} in Appendix~\ref{app:structures} [steps~4--5].
She then constructs the block payload $X$ using the \nonce, the block metadata containing the public keys $(\pkdh, \pksig, \pkvrf)$, the public application data, and the root \textsf{MTR} of the Merkle tree. She signs the payload $X$ and the pointer $\ptr$ to the previous block using the previous signing key $\sksig'$ (see Figure~\ref{fig:block_structure}) [step~6].
Finally, she puts the obtained block, $B = (X, \ptr, \sigma)$, into the content-addressable store [step~8].

\begin{figure*}[t]
\fbox{
\begin{minipage}[t]{.45\linewidth}
\begin{algorithmic}
\Procedure{ExtendChain}{\textsf{data}, \textsf{claims}, \acm, \textsf{keys}, \ptr, \textsf{store}}
\begin{enumerate}[leftmargin=1.6em]
    \item Randomly generate a $\lambda$-bit nonce $\nonce$. 
    \item For each claim $(l, m)$ in \textsf{claims}:
        \[r, h, \claimkey, \proofkey, (i, \claim) \gets \Call{EncClaim}{\skvrf, l, m, \nonce},\]
        Additionally, record each randomizer $r$.
    \item For each tuple $(\pkdh^{R}, l)$ in $\acm$, encode a capability:
        \[(i_\capab, \capab) \gets \Call{EncCap}{\skdh, \pkdh^{R}, l, h, \claimkey, \proofkey, \nonce}\]
    \item Construct a set $S$ containing the encoded claims and capabilities. 
    \item Build a unique-resolution key-value Merkle tree from the set of entries $S$: $\textsf{MTR} \gets \Call{BuildTree}{S, \textsf{store}}$. 
    \item Compute the block
        \[B \gets \left((X, \ptr), \sigma = \sgn(\sksig', (X, \ptr)) \right)\]
    where $X$ contains the \nonce, \textsf{MTR}, the block metadata (containing the public keys $\pkdh, \pksig, \pkvrf$), and the public application data.
    \item Put the block $B$ into the store using $\Call{Put}{B}$
    \item \Return $H(B)$.
\end{enumerate}
\EndProcedure
\end{algorithmic}
\end{minipage}\hspace{.04\linewidth}
\begin{minipage}[t]{.48\linewidth}
\begin{algorithmic}
\Procedure{GetClaim}{$\skdh, l, \ptr, \textsf{store}$}
\begin{enumerate}[leftmargin=1.6em]
    \item Get the block from the store: $B \gets \Call{Get}{\ptr}$.
    \item Retrieve owner's public keys ($\pkdh^O, \pkvrf^O, \pksig^O$), the block's nonce $\nonce$, and the block map hash \textsf{MTR} from block $B$.
    \item Compute the capability lookup key:
        \[i_\capab \gets H_3(s \parallel l \parallel \nonce),\]
        where $s$ is the shared secret $s = \derivesecret(\skdh, \pkdh^O)$.
    \item Get the encoded capability from the tree:
        \[\capab \gets \Call{QueryTree}{\textsf{MTR}, i_\capab, \textsf{store}},\]
    \item Obtain the claim lookup key $i$, the VRF hash $h$, the claim encryption key $\claimkey$, and the proof verification key~$\proofkey$:
        \[i, h, \claimkey, \proofkey \gets \Call{DecCap}{\skdh, \pkdh^O, l, \capab, \nonce}\]
    \item Get the encoded claim from the tree:
        \[c \gets \Call{QueryTree}{\textsf{MTR}, i, \textsf{store}}\]
    \item Decode $c$ and verify the correctness of the claim:
        \[m \gets \Call{DecClaim}{\pkvrf^O, l, h, \claimkey, \proofkey, c, \nonce}\]
    \item If any of the lookups failed, return $\none$. If the verification failed, return $\bot$. Otherwise, return $m$.
\end{enumerate}
\EndProcedure
\end{algorithmic}
\end{minipage}
}
\caption{Extending and querying \keychains}
\label{fig:extend_and_query}
\end{figure*}

\descr{Chain validation.} Readers must always validate that new blocks correctly extend the chain that they have previously seen. To do so, users run the  procedure \Call{ValidateBlocks}{}, see Figure~\ref{fig:validate_blocks}. The input to this procedure is a list of blocks $B_i$, where $B_0$ is the last validated block, and $B_1$ through $B_t$ are the new blocks to be validated. For each new block $B_i$ the reader first checks if the block includes all elements: the payload, signature, and the pointer [step~1]. Next, she retrieves the public key $\pksig$ from the preceding block $B_{i-1}$ and verifies the signature in the block $B_i$ [steps~2--3]. This verifies the authenticity of the chain. Finally, she verifies that the pointer in the block $B_i$ is a hash of the preceding block $B_{i-1}$, which verifies the integrity of the chain [step~4].

\begin{figure}[t]
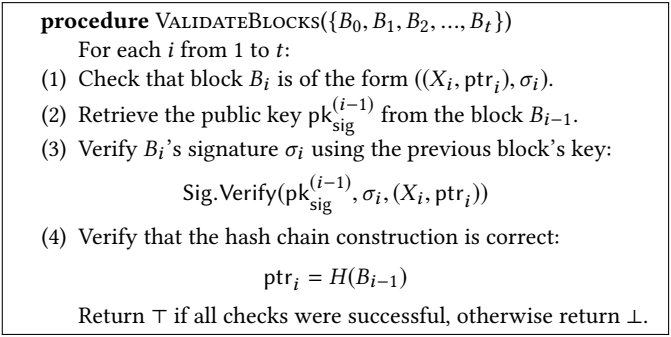

\fbox{
\begin{minipage}{\linewidth}

\begin{algorithmic}
\Procedure{ValidateBlocks}{$\{B_0, B_1, B_2, ..., B_t\}$}
\begin{enumerate}[leftmargin=1.6em]
    \item[] For each $i$ from $1$ to $t$:
    \item Check that block $B_i$ is of the form $((X_i, \ptr_i), \sigma_i)$.
    \item Retrieve the public key $\pksig^{(i-1)}$ from the block $B_{i-1}$.
    \item Verify $B_{i}$'s signature $\sigma_i$ using the previous block's key:
        \[ \sig.\vrfy(\pksig^{(i-1)}, \sigma_i, (X_i, \ptr_i))\]
    \item Verify that the hash chain construction is correct:
        \[ \ptr_i = H(B_{i-1}) \]
    \item[] Return $\top$ if all checks were successful, otherwise return $\bot$.
\end{enumerate}
\EndProcedure
\end{algorithmic}

\end{minipage}
}
\caption{Block validation}
\label{fig:validate_blocks}
\end{figure}

\descr{Retrieval of the claim by label.} After having validated the \keychain of an owner, the reader can query it to retrieve claims of interest using procedure \Call{GetClaim}{} in Figure \ref{fig:extend_and_query}. 
This procedure takes as input the reader's private Diffie-Hellman key $\skdh$, the claim label $l$, a pointer to the latest block $\ptr$ and the owner's store. The reader retrieves the block, and parses it to get the block's nonce, the owner's public keys, and the block map hash [steps~1--2]. She then derives the capability lookup key using the DH secret shared with the owner [step~3], queries the block map to retrieve the corresponding capability [step~4]. We refer to Algorithm~\ref{alg:query_tree} in Appendix~\ref{app:structures} for the details of the \Call{QueryTree}{} algorithm. Next, she runs the decoding procedure to obtain the claim lookup key $i$, the VRF hash $h$, the claim encryption key $\claimkey$, and the proof verification key $\proofkey$ [steps~5]. She then obtains the claim encoding $c$ by querying the tree with the claim's lookup key $i$ [step~6]. Finally, the reader decodes and verifies the encrypted claim using $h$, $\claimkey$, $\proofkey$ [step~7].



\medskip



\subsection{Security and privacy properties}
We now sketch why the \keychain design fulfills the security and privacy objectives established in Section~\ref{sec:problem_statement}.

We note that \emph{authenticity} and \emph{integrity} are guaranteed through the usage of signature and hash chains respectively. Signatures guarantee that the information stored in a \keychain has been added by the owner of the chain. The usage of cryptographic hash functions for constructing the pointers between blocks guarantees that tampering with the \keychain content will be detected. 

\descr{Privacy.} \Keychains provide \emph{privacy of content} and \emph{privacy of the social graph}. We capture these through the following properties:
\begin{itemize}[noitemsep,nolistsep,leftmargin=\parindent,label=--]
    \item \emph{Capability-reader unlinkability.} The adversary cannot determine for which honest user a capability has been created. 
    \item \emph{Claim privacy.} The adversary cannot learn anything about the labels and bodies of claims for which it does not have the corresponding capabilities. 
\end{itemize}

\smallskip

Informally, these properties are provided by \keychains because the adversary can neither derive the capability lookup key, nor learn the contents of the encoded capability without the knowledge of the shared secret used to encrypt the them. This implies that an adversary without this key cannot read capabilities nor learn to whom they are destined  (capability-reader unlinkability). Since the adversary cannot read the capability, it also  does not learn the VRF hash $h$ required to compute the claim lookup key, nor the claim encryption key $\claimkey$. Moreover, the pseudorandomness of the VRF hash $h$ ensures that the adversary cannot compute $h$ without the cooperation of the chain owner. Thus, the adversary cannot check whether a particular claim is included in the block.

Following a similar reasoning, the adversary cannot learn the content of a claim from its lookup key. Furthermore, the encoded claim $c$ does not reveal anything about the claim, except its length. Therefore, claim privacy holds, as long as all claims are of the same length, or padded to the same length.

We formalize these properties in Experiments~\ref{exp:claimPriv} and~\ref{exp:capReaderUnlink}, and prove them in Theorems~\ref{thm:claim_privacy} and~\ref{thm:cap_reader_unlink} in Appendix~\ref{app:privacy_proofs}.

\descr{Non-equivocation.} Our construction also \emph{prevents equivocation}. Specifically, it guarantees the following two properties:
\begin{itemize}[noitemsep,nolistsep,leftmargin=\parindent,label=--]
    \item \emph{Intra-block non-equivocation.} Within a given block, a \keychain owner cannot include two different bodies encrypted to different readers, having the same claim label.
    \item \emph{Detectable inter-block equivocation.} For any subset of \keychain blocks the owner can produce a proof that, for a given label $l$, all claims in these blocks belong to some set of allowed claims $M$ without revealing the claims themselves.
\end{itemize}
The latter property ensures that a user cannot selectively withdraw access rights between blocks to equivocate users. We detail this attack and the proof that mitigates it in Section~\ref{sec:inter_block_equivocation}.

\smallskip

The intra-block non-equivocation relies on three properties of the \keychain construction. First, the uniqueness of the VRF hash $h$ ensures that for a given label all readers will compute the same claim lookup key. Second, the unique-resolution property of our Merkle tree ensures that for a given lookup key all readers obtain the same claim encoding. Third, the claim commitment ensures that all readers will decrypt the same claim body.


We formalize both properties in Experiments~\ref{exp:block_non_eq} and~\ref{exp:deteq}, and prove them in Theorems~\ref{thm:block_non_eq} and~\ref{thm:deteq} respectively in Appendix~\ref{app:equiv_proofs}.

    \section{Using \keychains to secure in-band key distribution}
\label{sec:using}

Recall from Section \ref{sec:problem_statement} that the goal of the \keychain data structure is to improve the security and privacy of in-band key distribution. In this section we describe how this can be achieved.

\descr{Building a \keychain.} To use a \keychain, a user has to build blocks, containing her claims. 
When to update the \keychain depends on the owner's preferences. For example, a user can update her chain whenever she rotates her own encryption public key, or when she needs to distribute new cross-references that are not present in her \keychain yet.

To update her chain, an owner runs the \Call{ExtendChain}{} procedure (Figure~\ref{fig:extend_and_query}). 
For this purpose, she encodes a set of claims representing all her current views of other users as cross-references in the following way. For each contact, she makes a cross-claim $(l, m)$, where $l$ is the contacts' e-mail, and $m$ is the contact's latest block.

Then, the owner must decide which of these claims she intends to make available to which of her contacts. This choice determines the access control set $\acm$. The access control policy is governed by the user's privacy preferences. Defining these preferences is beyond the scope of this work.

Recall that to implement access control the owner uses shared DH secrets with each of the readers. Thus, the owner needs to complete a round-trip of messages with a contact before she can give this contact access to her claims. 

Finally, the owner puts her own public encryption key into the public application data section of the block. For our use case of in-band key distribution we assume that all keys are constant size. Hence blocks, and therefore claims, are constant size too. This ensures claim privacy even though the encryption scheme leaks the length of the plaintext.

\descr{Distributing \keychains.} To fulfill their purpose, \keychains must be made available to other users. For this, a user includes a content-addressable store containing blocks from her \keychain, and a subset of the Merkle tree nodes from her latest \keychain block, in every message she sends. The user keeps a record of which blocks they have sent to whom. To select the blocks to be sent, the sender checks her record, and includes all her \keychain blocks that the recipients of the current e-mail have not received yet.

The subset of the Merkle tree is selected to ensure that all information in the \keychain relevant to her message can be authenticated. More concretely, the sender produces resolution paths on the tree (see the \Call{GetIncPath}{} procedure in Algorithm~\ref{alg:query_tree} in Appendix~\ref{app:structures} for the details) for each relevant claim and capability.


\descr{Receiving messages and validating \keychains.} 
Upon receiving a message with a store containing \keychain data, a user first validates the received chain, running the \Call{ValidateBlocks}{} procedure (Figure~\ref{fig:validate_blocks}) to check if the new blocks extend a chain that has been seen previously. If the validation succeeds, the owner checks the \emph{consistency} of the cross-references in the newly received part of the chain, i.e., whether all the cross-references to Charlie point to the blocks on a single chain. This partially prevents malicious chain owners from cross-referencing fake chains. See Section~\ref{sec:inter_block_equivocation} for an example of such an attack, and the details of a consistency check procedure in case the receiver does not have access to the claim in some of the received blocks.
If both checks succeed, she stores all the received blocks and tree nodes into her \emph{gossip storage}. This enables her to query the sender's \keychain later. The gossip storage contains all the block and tree nodes the user has received over time. 

\descr{Message encryption.} 
Following the opportunistic encryption paradigm, before sending a message, the sender checks if she has learned the public keys of all of the recipients through the \keychains she has received over time. If she cannot find all keys, she sends the message in plaintext.

To find the encryption keys she proceeds as follows. For every recipient with e-mail address $l$, and every \keychain with head $\ptr$ in her gossip storage \textsf{gossip\_store}, she runs \Call{GetClaim}{$\skdh,\allowbreak l,\allowbreak \ptr,\allowbreak \textsf{gossip\_store}$}, see Figure~\ref{fig:extend_and_query}, to find out whether it includes cross-references to this recipient. For every hit, she parses the corresponding claim and adds the cross-referenced \keychain block of the recipient to a \emph{social evidence set} for this recipient. 
She then identifies the most recent block of the recipient's \keychain (out of those present in her evidence set), i.e., the one that forms the longest hash chain, and uses the encryption public key in that block to encrypt the message. 
As a result of this process, the sender may discover new blocks of the recipient's chain. She can then include the updated views as cross-references next time the chain is extended.

\descr{Resolving conflicts.}
This key resolution process may reveal conflicting views. For example, the blocks in the evidence set could point to two or more distinct chains. Another possibility is there could be a `fork': two valid blocks with the same block index that extend a common parent block. 
In either case, \keychains conflicts are detectable and generate cryptographically non-repudiable evidence. The design of mechanisms for sharing such evidence and deciding how to act on it is out of scope of this work. In Section~\ref{sec:simulations}, however, we empirically measure the number of distinct views that a sender would on average have about a recipient, to quantify if resolving conflicts is possible at all in a decentralized setting.

\subsection{Detecting inter-block equivocation}\label{sec:inter_block_equivocation}

\begin{figure}
    \centering
    \begin{tikzpicture}[%
    font=\small,
    block/.style={
      draw, rectangle,
      minimum height=2em,
      text centered,
      text width=2em,
      anchor=mid
    },
    xref/.style={->,dashed},
    link/.style={->},
    chainlabel/.style={align=center,font=\small\itshape}]
    \matrix (o) [matrix of math nodes,column sep=12mm,nodes={block}]{
    O_1 & O_2 & O_3 & O_4 \\
    };
    \matrix (f) [above=6mm of o,matrix of math nodes,column sep=15mm,nodes={block}]{
    C_1 & C_2 & C_3 \\
    };
    \matrix (c) [below=6mm of o,matrix of math nodes,column sep=25mm,nodes={block}]{
    F_1 & F_2\\
    };

    \node[left=0mm of o,chainlabel] (lo) {Owen's\\chain};
    \node[above=5mm of lo,chainlabel] {Real\\Charlie's\\chain};
    \node[below=5mm of lo,chainlabel] {Fake\\Charlie's\\chain};

    \draw[xref] (o-1-1) -- node[right=1mm,align=left] {xref to Charlie\\for Alice} (f-1-1);
    \draw[xref] (o-1-2) -- node[right=1mm,align=left] {xref to Charlie\\for Bob} (c-1-1);
    \draw[xref] (o-1-3) -- node[right=1mm,align=left] {xref to Charlie\\for Bob} (c-1-2);
    \draw[xref] (o-1-4) -- node[left=1mm,align=right] {xref to Charlie\\for Alice} (f-1-3);

    \draw[link] (o-1-2) -- (o-1-1);
    \draw[link] (o-1-3) -- (o-1-2);
    \draw[link] (o-1-4) -- (o-1-3);
    \draw[link] (f-1-2) -- (f-1-1);
    \draw[link] (f-1-3) -- (f-1-2);
    \draw[link] (c-1-2) -- (c-1-1);
    \end{tikzpicture}

    \caption{Inter-block equivocation}
    \label{fig:inter_block_equiv}
\end{figure}
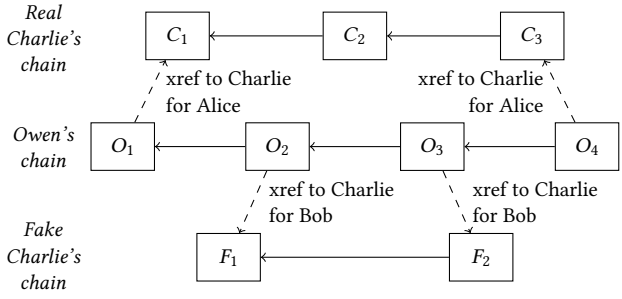

\Keychain's intra-block non-equivocation property ensures that all readers of cross-references to Charlie's chain see the same cross-reference in each block. However, chain owners may try to present different views to different users in different blocks by abusing the access-control mechanism. 
Thereby, the chain owner can equivocate between blocks.


Consider the following example, illustrated in Figure~\ref{fig:inter_block_equiv}, in which the chain owner Owen shows Bob a fake cross-reference to Charlie's chain, while showing the correct cross-reference to Alice. To do so, he never lets Alice and Bob see claims about Charlie's chain in the same block. In block 1, he gives access to Alice, but not to Bob, while in blocks 2 and 3, he gives access to Bob, but not to Alice. Finally, in block 4, Owen again gives access to Alice but not Bob. If Owen has claims about Charlie's true chain in blocks 1 and 4---the ones that Alice can read---and false claims about Charlie's chain in blocks 2 and 3---the ones Bob can read---he is effectively launching an equivocation attack.

A trivial solution to prevent this attack would be to, upon suspicion, allow Alice and Bob to inquire about claims related to Charlie in the blocks where they do not have access. However, this can leak information about if and when the chain owner learned about Charlie's updates. To be able to withdraw the access while preventing the described attack in a privacy-preserving way, \keychain enables the chain owner to prove, in zero knowledge, that she did not equivocate in the blocks where the cross-references were not accessible by the reader.

Consider again our example in Figure~\ref{fig:inter_block_equiv}. When Alice regains access to Charlie's references in block 4, she can use a detection mechanism to detect Owen's equivocation attempt. In other words, she can determine that in the intermediate blocks 2 and 3, where she did not have access to the cross-references about Charlie, Owen referenced a different chain than the one she sees. Bob would also detect the equivocation if he regains read access.

To enable detection, upon giving the access to Alice in block 4 again, Owen constructs a non-equivocation proof as follows.
\begin{enumerate}[leftmargin=1.6em]
    \item Owen recomputes the VRF hashes $h_i = \vrf.\eval(\vrftokenExpr_i)$ for all intermediate blocks, and computes proofs of correctness $\pi_h^{(i)}$:
    \begin{multline*}
        \pi_h^{(i)} = \textsf{SPK}\big\{(\skvrf) : \pkvrf = g^{\skvrf} \land\\ h_i = \vrf.\eval(\vrftokenExpr_i)\big\}()
    \end{multline*}
    Alice can use the VRF hashes to locate the cross-reference to Charlie in the intermediate blocks of Owen. The proofs $\pi_h^{(i)}$ confirm that she found the correct claims for Charlie's label $l$.
    \item Owen proves in zero-knowledge that $\claimcommit_i$ commits to one of the intermediate blocks $C_1, \ldots, C_t$ on Charlie's chain:
    \[\begin{aligned}
        \pi_\textsf{ref}^{(i)} = \textsf{SPK}\big\{& (r_i, x_i) : \claimcommit_i = \commit(r_i, x_i) \land \\
        &x_i \in \{ H_q(C_1), \ldots, H_q(C_t) \} \big\}().
    \end{aligned}\]
\end{enumerate}
Owen compiles all the tuples $(h_i, \pi_h^{(i)}, \pi_\textsf{ref}^{(i)})$ and sends them to Alice. Alice uses these tuples to check that each of the intermediate blocks belong the same chain of Charlie that she saw before. If Owen indeed equivocated as in the example, Alice can detect this, since the proof verification would have failed. A detailed description of this procedure is given in Figure~\ref{fig:consistency-proof}. 

\begin{figure}[t]
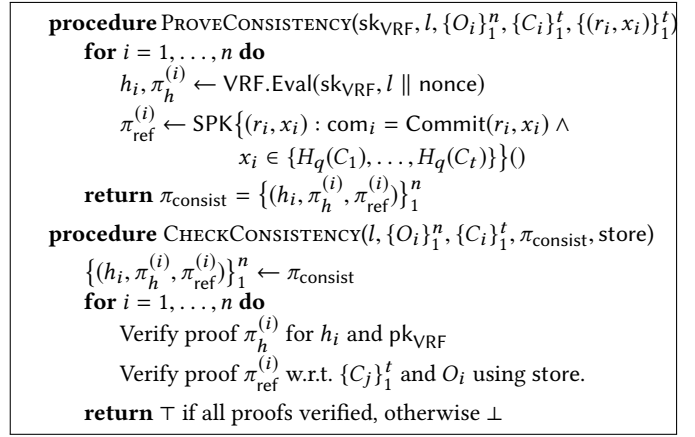

    \fbox{
    \begin{minipage}{\linewidth}
    
    \begin{algorithmic}
    \Procedure{ProveConsistency}{$\skvrf, l, \{O_i\}_{1}^{n}, \{C_i\}_{1}^{t}, \{(r_i, x_i)\}_1^{t}$}
    \For{$i = 1,\ldots,n$}
        \Let{$h_i, \pi_h^{(i)}$}{$\vrf.\eval(\vrftokenExpr)$}
        \Let{$\pi_{\textsf{ref}}^{(i)}$}{$\textrm{\textsf{SPK}}\big\{ (r_i, x_i) : \claimcommit_i = \commit(r_i, x_i) \land
        \newline\mbox{}\qquad\qquad\qquad\qquad
        x_i \in \{ H_q(C_1), \ldots, H_q(C_t) \} \big\}()$}
    \EndFor 
    \State \Return {$\pi_{\textsf{consist}} = \big\{(h_i, \pi_h^{(i)}, \pi_{\textsf{ref}}^{(i)})\big\}_1^{n}$}
    \EndProcedure
    
    \Procedure{CheckConsistency}{$l, \{O_i\}_{1}^{n}, \{C_i\}_{1}^{t}, \pi_{\textsf{consist}}, \textsf{store}$}
    \Let{$\big\{(h_i, \pi_h^{(i)}, \pi_{\textsf{ref}}^{(i)})\big\}_1^{n}$}{$\pi_{\textsf{consist}}$}
    \For{$i = 1,\ldots,n$}
        \State Verify proof $\pi_h^{(i)}$ for $h_i$ and $\pkvrf$
        \State Verify proof $\pi_{\textsf{ref}}^{(i)}$ w.r.t. $\{C_j\}_{1}^{t}$ and $O_i$ using \textsf{store}.
    \EndFor
    \State \Return {$\top$ if all proofs verified, otherwise $\bot$}
    \EndProcedure
    \end{algorithmic}
    \end{minipage}
    }
    \caption{Proving and verifying that blocks $O_i$ cross-reference the label $l$ to the correct chain $C_i$.}
    \label{fig:consistency-proof}
\end{figure}


    \section{Evaluation}
\label{sec:evaluation}



\descr{Experimental setup.} We implemented a prototype of \keychains in Python.\footnote{\url{https://github.com/claimchain/claimchain-core}} 
This implementation uses the petlib library~\cite{petlib} for elliptic curve cryptography operations, which internally relies on the OpenSSL C library. For the implementation of hash chains and unique-resolution Merkle trees we use the hippiehug\footnote{\url{https://github.com/gdanezis/rousseau-chain}} library, which is written in pure Python. Our implementation uses AES128 in GCM mode for symmetric encryption; ECDSA, ECDH, and other elliptic curve operations with a SECG curve over a 256 bit prime field (``secp256k1''); and SHA256 as the base hash function.
    
All the lookup keys on the claim map are truncated to 8 bytes, which makes collisions unlikely for up to $2^{32}$ entries in the map. The size of the per-block nonce is set to 16 bytes, and it is generated using the standard Linux \texttt{urandom} device. 

Our experiments are also publicly available and reproducible.\footnote{https://github.com/claimchain/claimchain-simulations} We extensively use Jupyter notebooks~\cite{Kluyver2016aa} and GNU parallel~\cite{Tange2011a}. We run the experiments on an Intel Core i7-7700 CPU @ 3.60GHz machine using CPython 3.5.2. 

\subsection{\keychain operations performance}
\label{sec:core-eval}
We now evaluate the performance of \keychains in terms of computation time and storage.

\descr{Timing.} We first measure the computation time for encoding and decoding claims and capabilities as described in Section~\ref{sec:claimdefinition}. We encode and decode 1000 claims and corresponding capabilities for random readers (i.e., encoded for a random DH public key). Each claim has a 32-byte random label and 512-byte random content. This reflects a realistic e-mail setting: 32-byte labels can accommodate e-mail addresses or their hash; and 512 bytes approximates the approx. 500-bytes block size in our experiments below.

\begin{table}[t]
\caption{\Keychain basic operations timing}
\label{table:core-eval}
\centering
\scalebox{0.85}{
\begin{tabular}{lccc}
{} &  mean (ms) &   std. (ms) \\
\hline
Label capab. lookup key computation  &      0.30 &  0.01 \\
Label capab. decoding                &      0.33 &  0.01 \\
Label capab. encoding                &      0.33 &  0.02 \\
Claim encoding [$\pi$ computation]   &      2.44 [2.38] &  0.05 [0.05]\\
Claim decoding [$\pi$ verification]  &      3.03 [2.96] &  0.05 [0.05] \\
\hline
\end{tabular}}
\end{table}

Table \ref{table:core-eval} reports our measurements. The time for encoding, decoding, and computing lookup keys for capabilities is under 0.33~ms. The time to encode and decode claims is around 3~ms, consisting mostly of the proof computation and verification time.

The most computationally expensive operation that \keychain owners perform is constructing the block map when a new block is created. The map is constructed using the \Call{BuildTree}{} procedure (see Algorithm~\ref{alg:build_tree} in Appendix~\ref{app:structures}). We measure the time to create a block map of $n$ claims with one capability each, i.e., readable by only one reader. We range $n$ from 100 to 5,000. For each case we construct a unique-resolution key-value Merkle tree with the encoded entries. Figure~\ref{fig:tree} (left) shows the average time required to build the tree across 20 experiments. Even for 5,000 claim-capability pairs the operation takes under 0.3 seconds. In reality, we expect users to have much fewer entries per block (in our simulation using the Enron dataset this number rarely exceeds 1,000).

\begin{figure}[t]
\centering
\begin{subfigure}[t]{.23\textwidth}
\centering
\includegraphics[width=\textwidth]{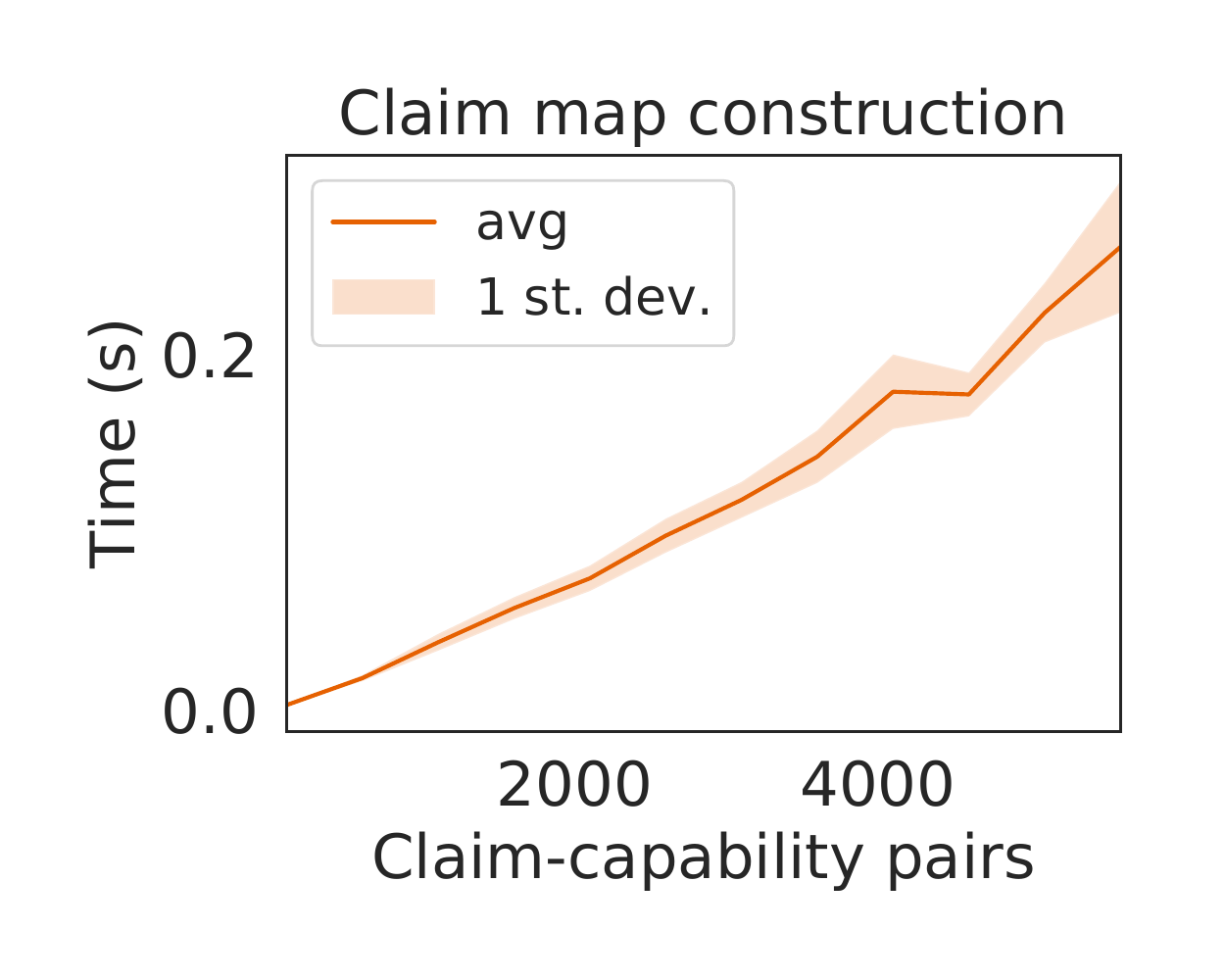}
\label{fig:tree-timing}
\end{subfigure}
\begin{subfigure}[t]{.23\textwidth}
\centering
\includegraphics[width=\textwidth]{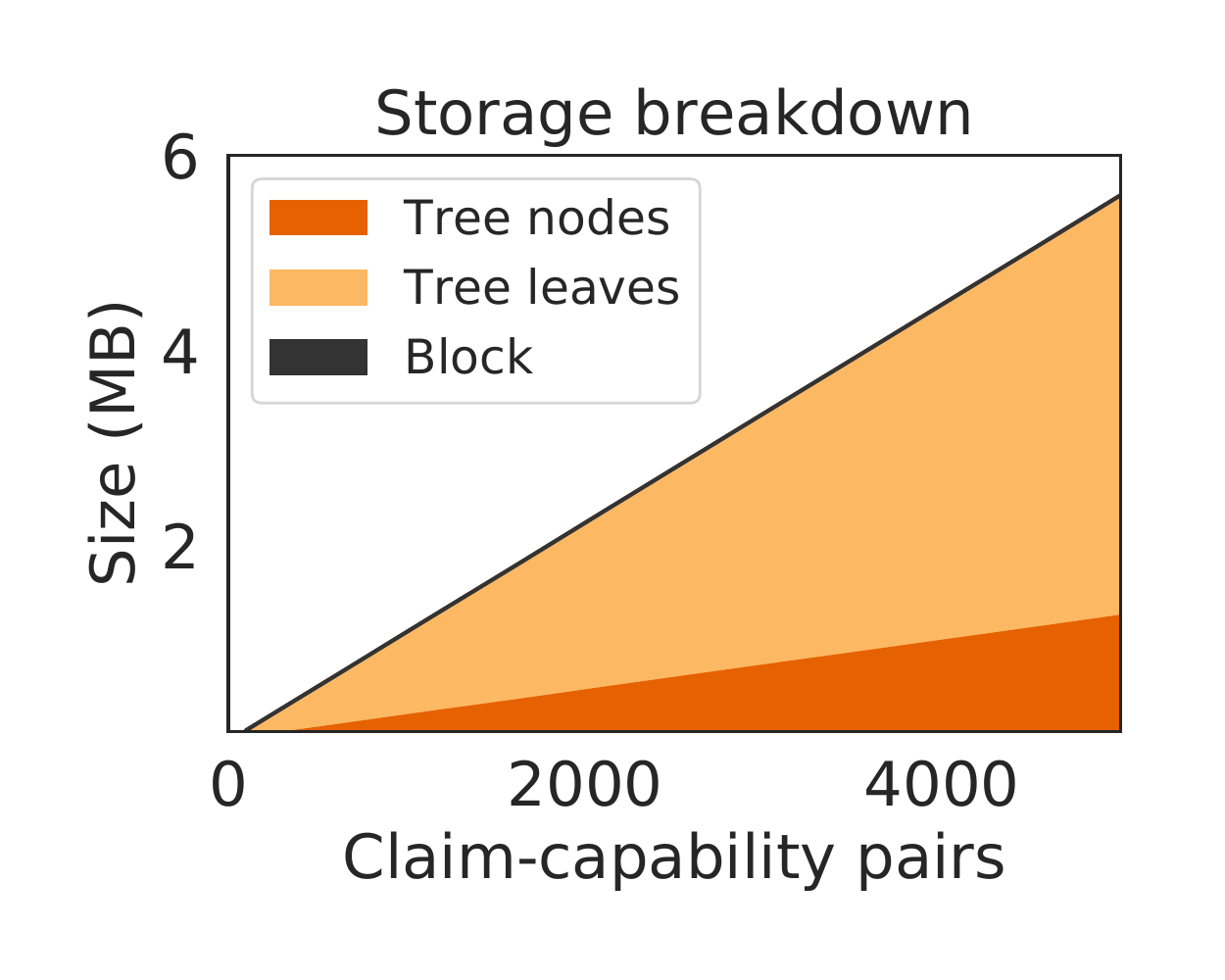}
\label{fig:storage-size}
\end{subfigure}
\vspace{-2em}
\caption{Total storage size and claim map construction time}
\label{fig:tree}
\end{figure}

Recall that along with the block, users send paths that prove the inclusion of relevant claims and capabilities in the block map tree. These are computed using the \Call{GetPath}{} procedure (see Algorithm~\ref{alg:query_tree} in Appendix~\ref{app:structures}). We measure the time to compute and verify a proof for a single entry, 
as well as the proof size in terms of number of tree nodes and bytes. We use the same setting as in the previous experiment. 
Unsurprisingly, the computation and verification time, and the proof size scale logarithmically with the number of items in the map. For 5,000 items, computation and verification take on average about 150 milliseconds, and the proof consists of on average 20 tree nodes and takes about 1.5 KB.

\descr{Storage.} We measure the size of a \keychain block, a block map tree, and values stored in the leaves of the tree (encrypted claims and capabilities). The size of the block map depends on the number of entries in the map and the size of claims. Figure~\ref{fig:tree} (right) shows the size breakdown depending on the number of items in the map. Note that the block itself only includes the root of the tree. Thus, the block size is constant (about 500 bytes), and can only grow if security parameters change (size of cryptographic public keys, or hash length increases), or additional data about the owner is added.

\descr{Inter-block equivocation detection.} The cost of proving consistency is dominated by the proof $\pi_\textsf{ref}^{(i)}$. Using a straightforward instantiation with `or' proofs, the prover and verifier must compute approximately $5t$ exponentiations to construct and verify  $\pi_\textsf{ref}^{(i)}$, where $t$ is the number of possible cross-referenced blocks. Therefore, a full consistency proof requires approximately $5nt$ exponentiations, where $n$ is the number of intermediate blocks on the owner's chain.
    \subsection{\keychain for in-band key distribution}
\label{sec:simulations}

In this section we study the use of \keychains for supporting in-band key distribution. We make use of the Enron dataset~\cite{KlimtY04,LeskovecLDM09} as a realistic test load to drive our experiments. It contains 500,000 e-mails of 147 Enron employees (230,000 after removing duplicates and non-readable e-mails). 

To simulate the use of \keychains, we loop through the e-mails in this dataset chronologically, updating \keychains of senders and receivers after each sent e-mail. For the experiments we run simulations of 10,000 consecutive e-mails from the full log of e-mails. We consider that at the beginning of a simulation the senders' \keychains are empty. During simulation they embed \keychain data in their messages as described in Section~\ref{sec:using}, and that upon receiving a message, users store all \keychain data locally.

We consider a scenario, called \emph{private setting}, in which senders selectively share the cross-references using \keychains. The rights to access claims are granted incrementally via `introductions'. Every time there is an e-mail with more than one recipient, all recipients earn a capability to access the cross-references about the correspondents in the sender's \keychain. Such capabilities are persistent over time. A user updates her chain when any of the recipients should receive a capability to a claim, and either the capability or the claim are not already in the \keychain.

As a baseline, we compare it to the \emph{public setting}, another scenario in which users gossip all the information available to them to their correspondents. They do not use ClaimChains since there is no need for access control. This scenario is close to the operation of the PGP Web of Trust if the users always were attaching their public key along with signatures on the keys of all their friends (cross-references). Since users share all the information they have available with everyone, the effectiveness of this setting reflects the fundamental limit for decentralized in-band key distribution.

To show how key propagation properties vary with the underlying social graph, we consider two groups of users from the dataset. First, only the 147 Enron employees, which represents a dense, well-connected, small social network with frequent messaging. Second, all users in the dataset, which is a sparsely connected network with many sporadic messages.

\begin{figure*}[t]
\centering
\includegraphics[width=0.49\textwidth]{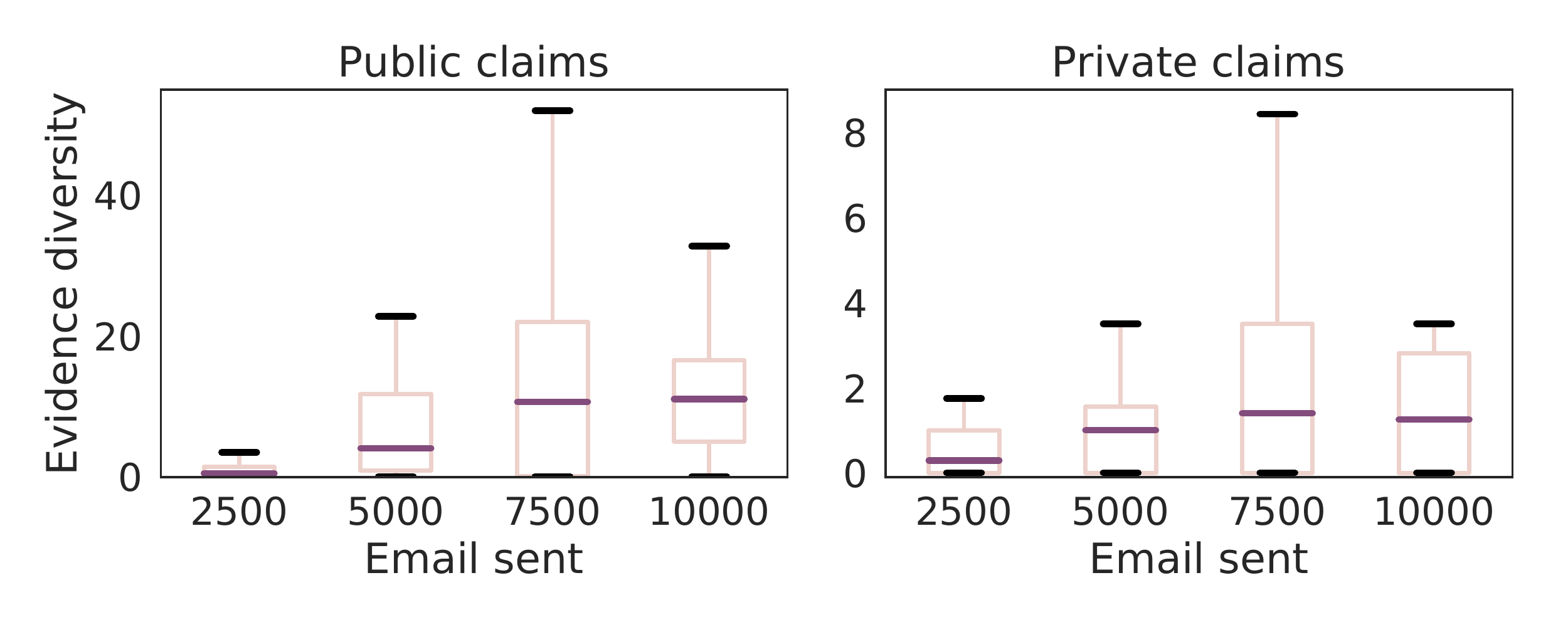} 
\includegraphics[width=0.49\textwidth]{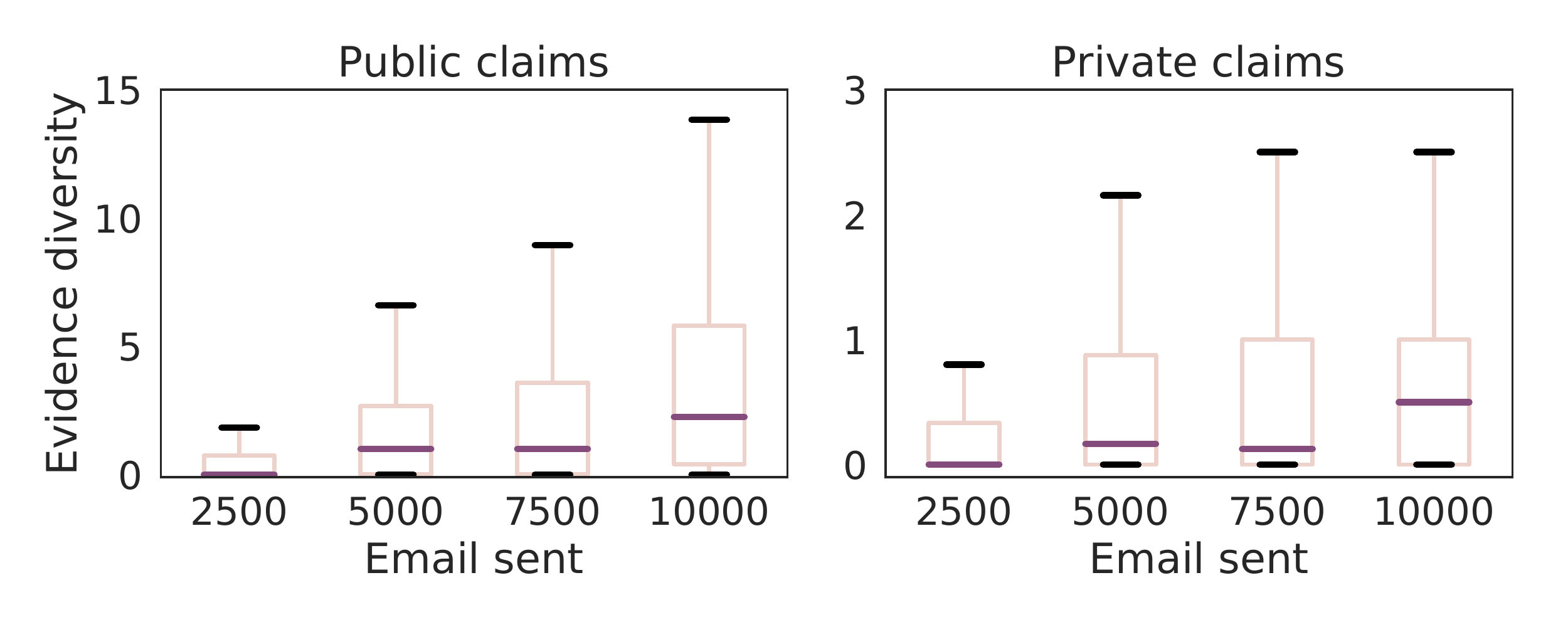}
\caption{Evidence diversity (within Enron, left, and all users, right).}
\label{fig:diversity}
\end{figure*}

\begin{figure*}
\includegraphics[width=0.32\textwidth]{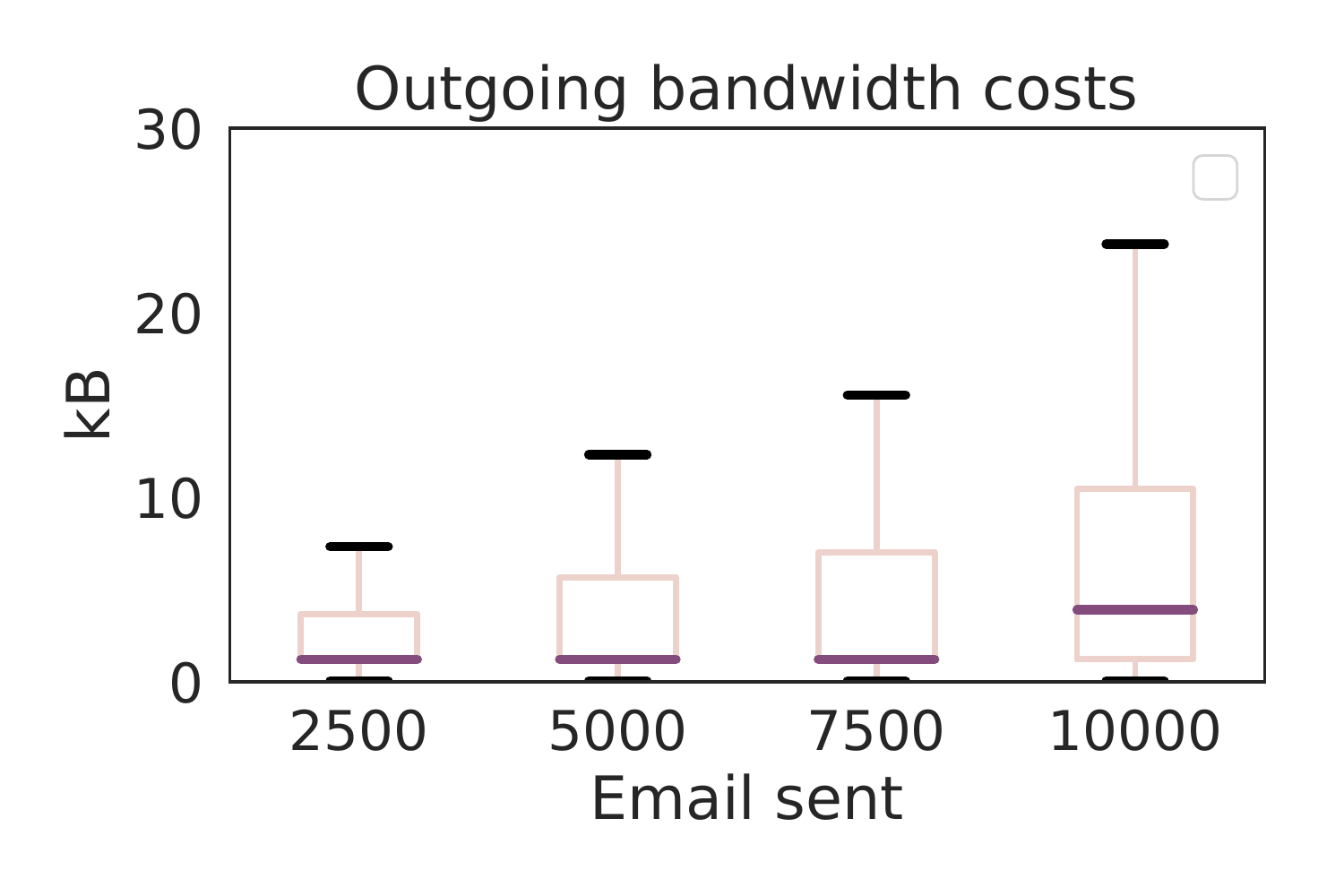} 
\includegraphics[width=0.32\textwidth]{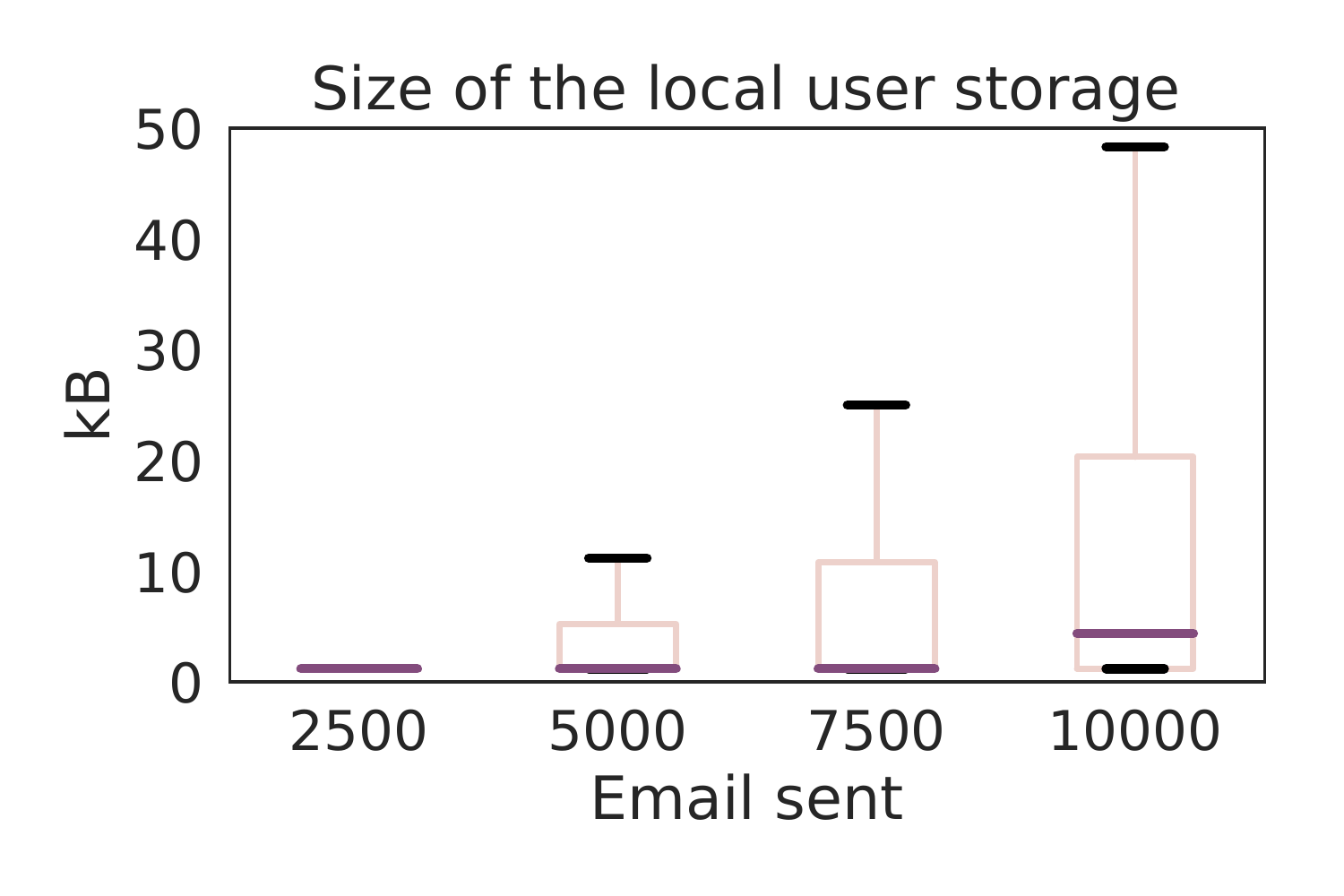} 
\includegraphics[width=0.32\textwidth]{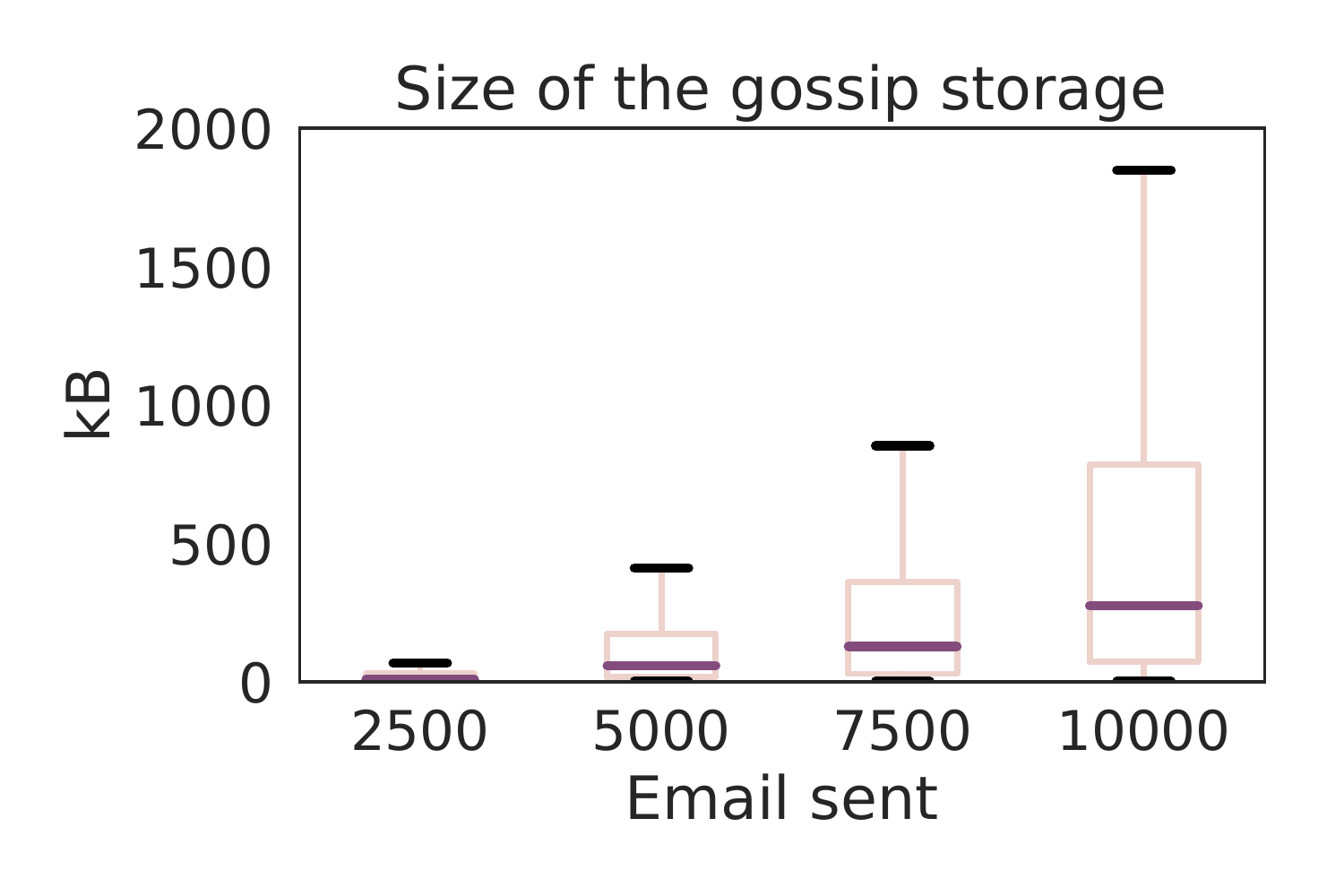} 
\caption{Storage and bandwidth measurements}
\label{fig:overhead}
\end{figure*}

\descr{Resilience to conflicts in views.}
\label{sec:resilience}
We first study how \keychains protect against attackers that aim at misleading users by reporting fake information about others. For each e-mail we record the amount of views that the sender has collected over time about each recipient. We call this quantity \emph{social evidence diversity}. When the evidence diversity is 1 for a given recipient, the sender only has one view---that of her own. When it is 10, that means the sender knows of 9 other people that have cross-referenced the recipient's chain. Intuitively, the higher is the diversity, the more users have to be corrupted by an adversary in order to convince the sender of a non-truthful \keychain state of the recipient.

Fig.~\ref{fig:diversity} illustrates the results in one batch of 10,000 messages. Results in other batches are similar. Unsurprisingly, in the public setting the amount of evidence is much higher than in the private setting, since much more information is exchanged. We also see that including all users reduces the mean diversity, since the social graph is sparse and many users do not have the opportunity to gather enough information about their correspondents.

Social evidence can also differ because over time views get outdated. In this case, the availability of the latest \keychain state of a user in our decentralized in-band setting is fundamentally limited by users' communication behavior. This is not a critical issue, since the construction of \keychains enables to differentiate `forks' resulting from an attack from simple chain updates. For example, in our private setting, social evidence includes views of the same chain at different time instants (on average in 1\% of cases among all users, and in 5\% within the Enron employees). In all these cases, \keychain enables to establish which view is the most up-to-date in a given social evidence set.

\descr{Storage and bandwidth costs.} To evaluate the overhead imposed by the use of \keychains for each user we record the size of the \keychain data being sent with each e-mail, and the required storage. We separately measure \emph{self-storage}---the space taken by users' own \keychain blocks and tree nodes, and \emph{gossip storage}---the space taken by information received from other users. 

The Figure~\ref{fig:overhead} shows that these costs rise over time as chains grow. We observe a large variation in growth caused by the variation in users' behaviour within the dataset. In the extreme, after 10,000 sent e-mails the required bandwidth per message is under 30~kB, the total size of the self-storage is under 50~kB, and the size of the gossip storage is under 2~MB. Note that we only report here the results in the private setting, since it employs \keychains.

\descr{Effectiveness of in-band key distribution.} Recall that the end goal of the public-key distribution is to enable end-to-end encrypted communication. Thus, we measure the effectiveness of the distribution as the percentage of encrypted e-mails, i.e., the fraction of times when a sender has received at least one \keychain block of \emph{all} the recipients, directly or through gossiping, and could find the key to encrypt the message to them. We simulate 10,000 e-mails starting at arbitrary points in the e-mail log and record whether senders have enough information to encrypt the e-mail. Figure~\ref{fig:enc_status} shows how this proportion changes over time. 

On the left we show the results within the set of Enron employees. At the beginning of the trace users learn many keys and increasingly more e-mails get encrypted. After some time the discovery rate decreases and there is a large variation in the proportion of encrypted e-mails. For this particular run, in the public setting, the overall proportion is 66\%, decreasing to 57\% in the private setting. 
We study the variance running the private setting simulations ten times at different points in the log. We observe that by the 2,000-th, 4,000-th, 6,000-th, and 8,000-th e-mail, on average 38\% ($\pm15$), 50\% ($\pm11$), 55\% ($\pm12$), 57\% ($\pm12$) of e-mails are encrypted\footnote{By $\pm x$ we denote 95\% Student t-distribution confidence interval in percentage points}. The overall percentage of encrypted e-mails across the ten runs is 59\% ($\pm9$). Note that these rates do not consider possible key gossiping that could have happened through users outside of the company.

On the right we consider all users. Compared to the previous measurements, the proportion of encrypted messages is significantly lower, overall average being 26\% in the public setting. This is because many senders and recipients are outside of the group of Enron employees and thus exchange fewer, or sporadic, e-mails. In the private setting the proportion decreases by only 4~p.p to 22\%. Across all ten runs, the overall proportion is 23\% ($\pm7$).
Note that these numbers are a lower bound on the number of encrypted e-mails. Since we do not observe the inbox of non-Enron users, we cannot establish the effectiveness of gossiping. Thus, the propagation may be better than shown in our measurements.

\begin{figure*}[t]
\includegraphics[width=0.49\textwidth]{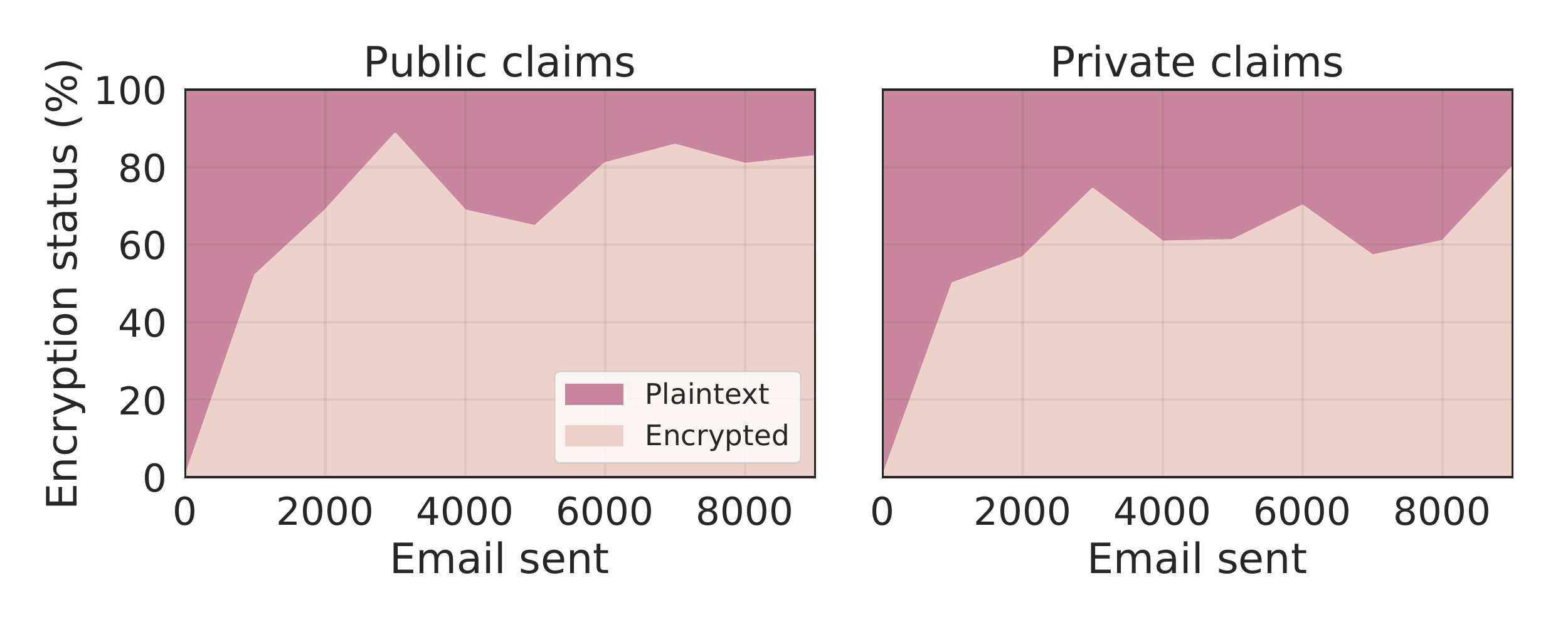}
\includegraphics[width=0.49\textwidth]{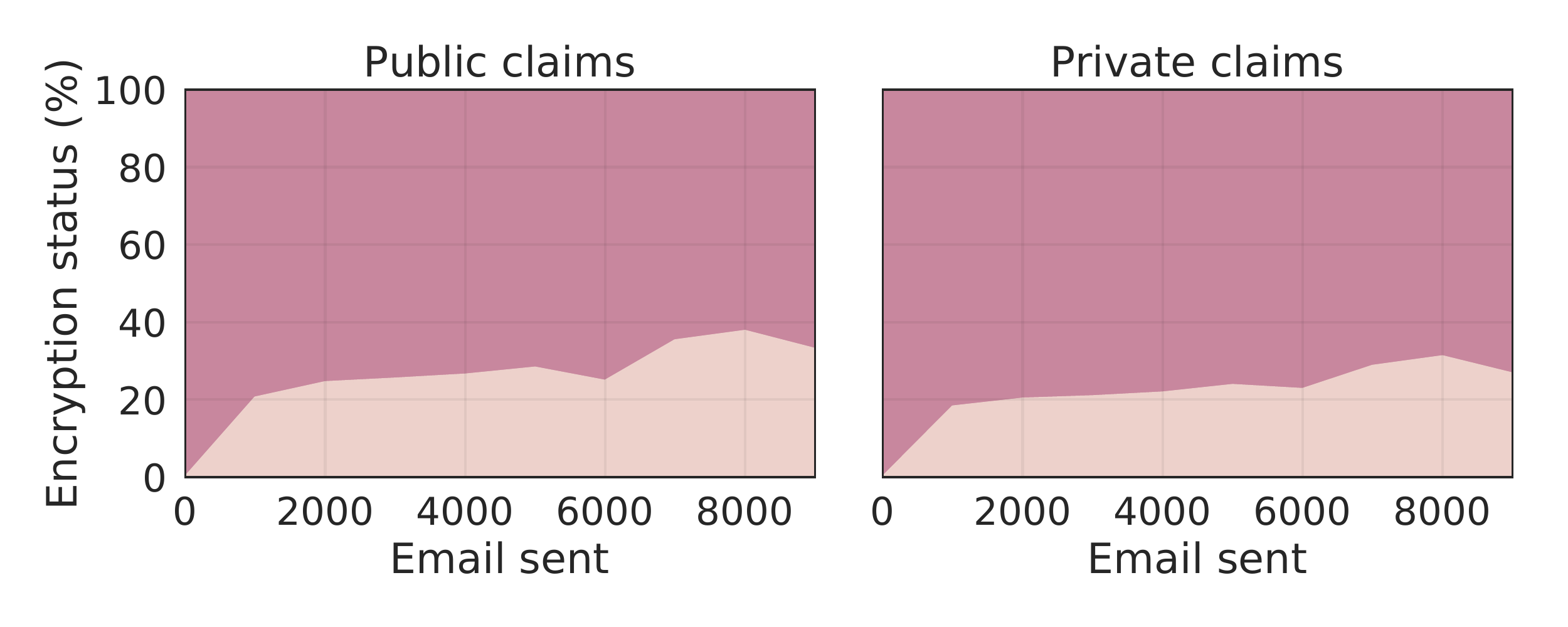}
\caption{Encryption status of e-mails (within Enron, left, and all users, right). Proportions are computed over groups of 1,000 consecutive e-mails.}
\label{fig:enc_status}
\end{figure*}

\descr{Takeaways.}
As expected, promiscuous public gossiping is more effective at propagating the key information than privacy-preserving sharing. Nonetheless, its advantage is relatively small. Sacrificing users' privacy does not provide a significant increase in the proportion of encrypted e-mails. This suggests that selective revealing of cross-references, enabled by \keychain cryptographic mechanisms, can offer a better trade-off between privacy and utility than the traditional Web of Trust-like sharing model.

On the other hand, the gain in privacy comes at a cost in resistance to active attacks. Even though gossiping in the private setting does not significantly decrease the proportion of encrypted e-mails, it does significantly deteriorate the resilience to malicious users. The evidence diversity in the private setting can be up to 10 times lower than in the public setting. This means that on average fewer users need to be compromised to propagate inauthentic keys.

The public setting in our simulations represents an upper bound for key propagation as in this setting users share all the information available to them. Our results corroborate that, independently of the use of \keychains to secure cryptographic material, in-band key distribution is unlikely to achieve full coverage, and furthermore the coverage is largely unpredictable. In a decentralized setting, key propagation cannot be more effective unless additional communication channels are employed. 

    \section{Comparison with existing systems}
\label{section:related_work}

In this section we compare \keychain to other key distribution systems targeting e-mail communications. Designs that facilitate secure key distribution in other contexts, like Certificate Transparency~\cite{rfc6962} for HTTPS connections, are out of scope.

We consider four approaches that represent existing deployed and academic solutions. First, we consider the PGP in-band decentralized approach, where users attach their public keys in outgoing e-mails. This corresponds to the current implementation of Autocrypt where there is no gossiping of contacts' keys. Second, we compare \keychain to systems that employ highly available centralized key servers taking the SKS Keyservers, a pool of synchronized servers that store PGP keys, as reference. Third, we consider solutions that, like \keychain, use high-integrity data structures (hash chains and Merkle trees) so as to hold providers accountable for the bindings they serve. In this space we consider CONIKS~\cite{MelaraBBFF15}, a federated variant, 
and Keybase\footnote{\url{https://keybase.io}}, a centralized design. 
Finally, we consider approaches such as Namecoin\footnote{\url{https://namecoin.org}} that bind identities to cryptographic wallet addresses on Proof-of-Work blockchains~\cite{nakamoto2008bitcoin}.

\begin{table*}[t]

\newcommand{\tddag}{\tnote{$\ddagger$}}
\newcommand{\tdag}{\tnote{$\dagger$}}

\centering
\caption{Comparison of key distribution systems from an end-user perspective. 
\textit{Social graph visibility:} who learns the user's social graph. \textit{Active attack detection:} whether active attacks by malicious providers, users, and network adversaries can be detected. \textit{Total key availability:} guarantee that recipients' current encryption keys are always available to senders. \\
}
\label{tab:comparison}
\begin{threeparttable}
\vspace{-2em}
%
%
%
\scalebox{0.85}{
\begin{tabular}{l|ll|lll|l}
                                    &   \textbf{In-band PGP}   & \textbf{SKS Keyserv.}& \textbf{CONIKS}    
                                    & \textbf{Keybase}          & \textbf{Namecoin}     & \textbf{\keychain} \\

\textbf{Social graph visibility}            &  E-mail provider & Public            & Provider                          & Public                            & Public                            & Authorized readers \\
\textbf{Active attack detection}            &   \xmark              & \xmark            & \cmark\tddag                      & \cmark\tddag                      & \cmark\tddag                      & \cmark \\
\textbf{Total key availability}             &   \xmark              & \cmark            & \cmark                            & \cmark                            & \cmark                            & \xmark                \\
\hline
\textbf{Sending bandwidth, $O(\cdot)$}      &   $s \cdot b$         & $s \cdot b^2$       & $s \cdot b \cdot \log(n)\tdag$    & $s \cdot b \cdot (b + \log(n))\tdag$  & $s \cdot b \cdot (b + \log(n))\tdag$      & $s \cdot b^2 \cdot \log(b)$ \\
\textbf{Receiving bandwidth, $O(\cdot)$}    &   $r \cdot b$         & $r \cdot b$       & $r \cdot \log(n)\tdag$            & $r \cdot ( b +  \log(n))\tdag$        & $r \cdot (b + \log(n))\tdag$              & $r \cdot b^2 \cdot \log(b)$ \\
\textbf{Local storage, $O(\cdot)$}          &   $b^2$               & $b^2$             & $b + \log(n)$                     & $b^2$                                 & $b^2$                                     & $r \cdot b^2 \cdot \log(b)$
\end{tabular}
}

\vspace{.5em}

\begin{tablenotes}
\item[$\dagger$] Without costs for auditing the transparency log / verifying blockchain history   \quad $\ddagger$ Requires global consensus on system's state
\end{tablenotes}

\end{threeparttable}
\end{table*}


We compare these systems in terms of functionality, computational and network costs. The result is summarized in Table~\ref{tab:comparison} where we use the following notation: $n$---number of users, $s$---number of sent messages, $r$---number of received messages, $b$---maximum total number of contacts of any user after $s$ sent and $r$ received messages.

\descr{PGP key distribution.}
OpenPGP~\cite{rfc4880} is a format for sharing PGP key material. It enables users to vouch for each others' keys, constructing a Web of Trust.
In-band PGP key distribution, whereby users' keys are embedded in messages, is vulnerable to attacks by malicious e-mail providers and network adversaries who, first, gain complete visibility of the users' Web of Trust, and second, can replace the attached keys and in that way compromise the confidentiality of users' communications. The Autocrypt implementation protects from servers launching man-in-the-middle attacks, but still reveals the contacts. Moreover, malicious users can equivocate by sharing different versions of others' keys with different readers. \keychain can be seen as an alternative format to OpenPGP where cryptographic mechanisms hide contacts' information. Furthermore, it ensures that all readers retrieve the same cross-reference for a contact, and makes past equivocation attempts detectable.

For in-band PGP key distribution users need to attach their keys and cross-references in their e-mails, hence requiring $O(s \cdot b)$ and $O(r \cdot b)$ outgoing and incoming bandwidth respectively. Locally, they need to store the keys of all their friends ($O(b)$), and their friends' cross-references ($O(b)$ for each friend), resulting in $O(b^2)$ cost. 
As shown in Section~\ref{sec:evaluation}, key propagation efficiency in this setting strongly depends on the social graph and on the users' behavior, and is not constant.


Centralized PGP PKI providers allow to achieve 100\% key availability, but introduce security and privacy concerns. The most widely deployed implementation, the SKS Keyservers, does not defend against malicious providers that serve fake keys, or that exploit key lookup requests to learn users' relationships. In addition, 
they accept unauthenticated plaintext HTTP requests, and thus network adversaries can also perform these attacks.

Centralized PKIs obviate the need for attaching key material in outgoing e-mails. However, before sending a message, a user needs to obtain the latest key of each of the recipients (which contains their cross-references, $O(b)$), requiring $O(s \cdot b^2)$ bandwidth. Upon receiving a message, users look up the sender's key (containing $O(b)$ cross-references), thus also resulting in $O(r \cdot b)$ network cost. As in the previous setting, users store the keys of their friends, including their cross-references, requiring a storage of $O(b^2)$.
%

\descr{Accountable key repositories.}
Recent approaches to key distribution rely on transparency logs~\cite{MelaraBBFF15} to prevent active attacks from malicious providers and network adversaries. 


In CONIKS~\cite{MelaraBBFF15}, providers maintain cryptographically signed hash chains that can be audited for serving the correct key bindings for their users.
Key lookup responses in CONIKS include a Merkle tree proof of inclusion (size $\log(n)$). Hence, considering that users make a lookup request for each of the recipients when sending, and for the sender when receiving, the bandwidth cost is $O(s \cdot b \cdot \log(n))$ for sending and $O(r \cdot \log(n))$ for receiving. Users store the keys of the friends, and the current inclusion proof for their own key, which requires storage space bound by $O(b + \log(n))$. We do not consider bandwidth costs related to verifying the history of the CONIKS provider or the consistency of users' keys.


Keybase maintains a global auditable hash chain that contains commits to users' individual \emph{sigchains} through a global Merkle tree. These sigchains are self-signed objects that evolve over time and include information about owner's keys, devices, online profiles, and friends. The global chain is occasionally cross-referenced into the Bitcoin's blockchain guaranteeing that it cannot be tampered with. Keybase users can create cross-references to their contacts by adding a snapshot of their contacts' state into their own sigchain. These cross-references are public, thus reveal user relationships.  

The latest state of a Keybase user includes a Merkle proof in the Keybase tree and the cross-references of her friends, which is bound by $O(b + \log(n))$. Hence, the bandwidth cost for sending messages, assuming that senders retrieve the receivers' latest keys from Keybase, is $O(s \cdot b \cdot (b + \log(n))$. In the same way, the cost of receiving e-mails, including looking up the latest version of the sender key, is $O(r \cdot (b + \log(n))$. These estimations do not consider the cost for validating the Keybase history, the consistency of user history, or the Keybase cross-references on the Bitcoin blockchain. Users store locally the keys of their friends and their respective cross-references, thus requiring $O(b^2)$ storage. 

Unlike \keychains, neither CONIKS nor Keybase provide mechanisms to enable social verification while protecting the social graph of users. Furthermore, CONIKS, a federated design, and Keybase, a centralized design, both put providers in a privileged position to observe their users' communications patterns. \keychains are designed to work in a decentralized setting that does not rely on the existence of a single entity with a global view of the system.

\descr{PKIs inspired by the Bitcoin blockchain.} 
Some PKI systems leverage permissionless decentralized cryptographic ledgers. Besides high availability and resistance to 
tampering attempts, ledgers also provide a global namespace and mechanisms for achieving global consensus. 
For instance, Namecoin uses a proof-of-work blockchain~\cite{nakamoto2008bitcoin} to implement key discovery for secure messaging apps. In Namecoin key material is encoded in the OpenPGP format, thus revealing the users' social graph. Moreover, the public transactions reveal information such as how social graphs evolve, or how pseudo-identities are linked to the same owner. \keychains protect this information by means of cryptographic access controls.
Furthermore, blockchain-based systems involve a fee for obtaining and managing key material.
Instead of maintaining a global state, in \keychains each user controls a personal chain. Thus, there is no need for mining blocks and thus there are no costs for key management operations. 

From the perspective of end-users, the bandwidth and storage costs in Namecoin are the same as in Keybase. In a similar way, we do not consider costs involved in maintaining consensus in a global proof-of-work blockchain, or in verifying its history.

\descr{In-band key distribution with \keychains.}
Finally, we outline the costs involved with using \keychains. Besides her own information, owners include a claim for each of her cross-references and capabilities, along with the corresponding proofs of inclusion in the block map. Therefore, the size of a \keychain block is bound by $O(b^2 \cdot \log(b))$. Consequently, given that all messages attach a \keychain block, the total sending and receiving bandwidth cost is $O(s \cdot b^2 \cdot \log(n))$ and $O(r \cdot b^2 \cdot \log(n))$, respectively. Users are expected to store all \keychain blocks they receive, requiring a total storage of $O(r \cdot b^2 \cdot \log(n)).$

    \section{Concluding remarks}
\label{sec:discussion}

In-band key distribution, as proposed by Autocrypt, is a promising direction towards achieving e-mail encryption without the collaboration from service providers. However, it suffers from security and privacy problems. To address these issues we introduced \Keychains, a construction that can be sent in-band with e-mails to provide high-integrity evidence of key-identity bindings. Its cryptographic access control enables users to selectively reveal their contacts, preserving their privacy, while preventing equivocation attacks in which different users are shown different bindings. 

We demonstrate that key propagation, and thus the ability to encrypt messages, is not affected much when using the privacy features of \keychains. However, users do obtain less evidence about other users' bindings, increasing the chances that wrong keys go unnoticed.
On the negative side, our study shows that the coverage achieved by in-band key distribution is partial at best. In our realistic simulations we could achieve a maximum of 66\% of e-mail encrypted, even within a well-connected social network. 

However, we note that the design of \keychains is not tied to decentralized storage and distribution. Their strong security and privacy properties permit to host the content-addressable storage in semi-trusted providers without relying on them to return correct values. Such deployment of \keychains would greatly improve availability of \keychain data. But, to obtain perfect privacy, such scheme requires integration with privacy-preserving storage access~\cite{ChorGKS95,ToledoDG16} to avoid leakage stemming from access patterns. 

Finally, \keychain or its component data structures can have applications to use cases beyond key distribution. The claim map data structure, for example, can be applied in similar settings when a verifiable dictionary with cryptographic access controls for its lookup keys is needed.

    \section*{Acknowledgments}
    This research is funded by NEXTLEAP project\footnote{\url{https://nextleap.eu}} within the European Union's Horizon 2020 Framework Program for Research and Innovation (H2020-ICT-2015, ICT-10-2015) under grant agreement 688722. We thank Holger Krekel, Azul, and Harry Halpin for their feedback and discussions.
    
    \bibliographystyle{ACM-Reference-Format}
    \bibliography{main.bib}
    
    \appendix
    \section{Unique-resolution key-value Merkle tree}\label{app:structures}

Our unique-resolution key-value Merkle tree data structure is composed of two types of nodes:
\begin{align*}
\textsf{Internal} & = (\textsf{pivot},~\textsf{left}: H(\mathsf{Node}),~\textsf{right}: H(\mathsf{Node})) \\
\textsf{Leaf} & = (\textsf{key}, \textsf{value})
\end{align*}
We denote the root of a tree as \textsf{MTR}. Each \textsf{Internal} node contains a $\mathsf{pivot}$ string and the hashes of its two children. The invariant of the structure is that any nodes in the $\textsf{left}$ sub-tree will have pivots or leaf keys smaller than the parent pivot, and any nodes to the \textsf{right} sub-tree have pivots or leaf keys equal or larger than the parent pivot. As in a normal Merkle tree, the hash of the root node is a succinct authenticator committing to the full sub-tree (subject to the security of the hash function).

A proof of inclusion, or authentication proof, of a key-value pair in the tree involves disclosing the full resolution path of nodes from the root of the tree to the sought leaf. We show that such path is indeed a proof of inclusion, and, moreover, is unique in Section~\ref{sec:unique_resolution}.

\subsection{Algorithms}

\paragraph{Building the tree}

To build a tree from a set of key-value pairs $S = \{..., (k_i, v_i), ...\}$ we run the \textsc{BuildTree} procedure (Algorithm~\ref{alg:build_tree}) The procedure take as input a set of claims $S$ and a content-addressable \textsf{store}. It constructs the tree nodes and saves them to the \textsf{store}. Finally, it returns the hash of the root node of the resulting tree.

\begin{algorithm}[t]
\caption{Tree construction}
\label{alg:build_tree}
\begin{algorithmic}
\Procedure{BuildTree}{$S$, \textsf{store}}
    \If{$|S| = 1$}
        \Let{$\{(k, v)\}$}{$S$}
        \Let{\textsf{leaf}}{\textsf{Leaf}($k, H(v)$)}
        \State \Call{Put}{\textsf{store}, \textsf{leaf}}
        \State \Call{Put}{\textsf{store}, $v$} \Comment{{\footnotesize Put the value itself into the store}}
        \State \Return{$H(\textsf{leaf})$}
    \Else
        \State $(k^*, v^*) \gets_{\$} S$ \Comment{{\footnotesize Pick the pivot arbitrarily}}
        \Let{$(S^-, S^+)$}{\Call{Partition}{$k^*, S$}}
        \Let{\textsf{left}}{\Call{BuildTree}{$S^-, \textsf{store}$}}
        \Let{\textsf{right}}{\Call{BuildTree}{$S^+, \textsf{store}$}}
        \Let{\textsf{node}}{\textsf{Internal}($k^*, \textsf{left}, \textsf{right}$)}
        \State \textsf{store}.\Call{Put}{\textsf{node}} 
        \State \Return{$H(\textsf{node})$}
    \EndIf
\EndProcedure
\\
\Procedure{Partition}{$k^*, S$}
    \Let{$S^-, S^+$}{$\{~\}, \{~\}$}
    \For{$(k, v)$ in $S$}
        \If{$k < k^*$} \Comment{{\footnotesize Lexicographic comparison of strings}}
            \Let{$S^-$}{$S^- \cup \{(k, v)\}$}
        \Else{}
            \Let{$S^+$}{$S^+ \cup \{(k, v)\}$}
        \EndIf
    \EndFor
    \State \Return{$(S^-, S^+)$}
\EndProcedure
\end{algorithmic}
\end{algorithm}

\paragraph{Querying the tree}
The tree querying procedure \textsc{QueryTree} is described in Algorithm~\ref{alg:query_tree}. It takes as input the tree root $\textsf{MTR}$ and \textsf{store} that contains the tree nodes. The procedure traverses the tree starting from the root. For each intermediate node, the procedures follows a left or right sub-tree depending on the pivot field. It continues until it ends up in a leaf node. If the leaf node has the correct key, \textsc{QueryTree} returns the corresponding value, otherwise it returns $\bot$.

\begin{algorithm}[h]
\caption{Querying the tree}
\label{alg:query_tree}
\begin{algorithmic}
\Procedure{QueryTree}{\textsf{MTR}, $k$, \textsf{store}}
    \Let{$\pi$}{\Call{GetPath}{\textsf{MTR}, $k$, \textsf{store}}}
    \Let{$[..., \textsf{Leaf}(k', v)]$}{$\pi$}
    \If {$k' = k$}
        \State \Return {$\bot$}
    \Else 
        \State \Return {$v$}
    \EndIf
\EndProcedure
\\
\Procedure{GetPath}{$h$, $k$, \textsf{store}}
    \Let{\textsf{node}}{\textsf{store}.\Call{Get}{$h$}}
    \If{\textsf{node} is \textsf{Leaf}}
        \State \Return $[\textsf{node}]$
    \ElsIf{\textsf{node} is \textsf{Internal}(\textsf{pivot}, \textsf{left}, \textsf{right})}
        \If{$k < \textsf{pivot}$}
            \Let{$\pi$}{\Call{GetPath}{\textsf{left}, $k$, \textsf{store}}}
        \Else
            \Let{$\pi$}{\Call{GetPath}{\textsf{right}, $k$, \textsf{store}}}
        \EndIf
        \State \Return  $[\textsf{node}] + \pi$ \Comment{{\footnotesize Prepend the current node to the list $\pi$}}
    \EndIf
\EndProcedure
\end{algorithmic}
\end{algorithm}


\subsection{Unique resolution}\label{sec:unique_resolution}
For a given key, only one value can be stored in the tree. Any violation of this invariant will be detected when the tree is queried---thus the creator of the tree does not need to be trusted to enforce this invariant. More formally, for a given key $k$ it is only possible to successfully prove the inclusion of one leaf node in the tree with root $\textsf{MTR}$. We capture this notion in the \textsf{UniqRes} game in Experiment~\ref{exp:unique_resolution}. The following theorem states that no adversary can win this game.

\begin{theorem}\label{thm:unique_resolution}
For any probabilistic polynomial time adversary $\adv$ it holds that $\Pr_{\adv}[b = 1] = \negl(\lambda)$, where the bit $b \in \{0, 1\}$ is the output of the \textsf{UniqRes} game (Experiment~\ref{exp:unique_resolution}).
\begin{experiment}
\begin{algorithmic}
    \ExpTitle{UniqRes}
    \Let{$\mathsf{MTR}, k, \textsf{store}, \textsf{store}'$}{$\adv()$}
    \If {$\textsf{store} = \textsf{store}'$}
        \State \Return 0
    \EndIf
    \Let{$v$}{\Call{QueryTree}{$\mathsf{MTR}, k, \textsf{store}$}}
    \Let{$v'$}{\Call{QueryTree}{$\mathsf{MTR}, k, \textsf{store}'$}}
    \Let{$b$}{$v \neq v'$}
    \State \Return $b$.
\end{algorithmic}
\caption{Unique Resolution}
\label{exp:unique_resolution}
\end{experiment}
\end{theorem}
\begin{proof}
Assume $\adv$ wins the game. Then it is able to construct two stores such that there are two different valid paths:
\begin{align*}
\pi  &\gets \Call{GetPath}{\textsf{MTR}, k, \textsf{store}} \\
\pi' &\gets \Call{GetPath}{\textsf{MTR}, k, \textsf{store}'},
\end{align*}
that start with the same root $\mathsf{MTR}$, but end with different leaves containing $(k, v)$ and $(k, v')$ respectively.

First, assume one of the paths, w.l.o.g. $\pi$, consists of a single leaf node $t$ with $(k, v)$. Then the other path $\pi'$ can contain either another leaf $t'$ with $(k', v)$, or start with an internal node $t'$. This implies a hash collision, since $t \neq t'$, but $\textsf{MTR} = H(t) = H(t')$. By the collision resistance property of the cryptographic hash function $H$, this happens with negligible probability.

Now, assume that the paths have a common beginning. Let $t, t'$ be the first nodes along the paths that differ, and let $t^* = \textsf{Internal}(p^*, h^*_l, h^*_r)$ be their common parent. Then, there are four possible options:

\begin{enumerate}[label={(\alph*)}]
    \item Both $t$ and $t'$ are a left child of $t^*$. In this case, $H(t) = H(t') = h^*_l$. This implies a hash collision, which we assume to happen with negligible probability.
    \item Both $t$ and $t'$ are a right child of $t^*$. This is analogous to the previous case.
    \item The children $t$ and $t'$ are respectively the left child and the right child of $t^*$. This situation cannot happen, because \textsf{GetPath} decides which child to follow based on the value of the pivot $p^*$ and the lookup key $k$. Since the parent is common, the procedure will always choose either the left, or the right child.
    \item The children $t$ and $t'$ are respectively the right child and the left child of $t^*$. This is analogous to the previous case.
\end{enumerate}

Thus, the probability that $\adv$ wins the game, $\Pr_{\adv}[b = 1]$, equals the probability of a hash collision and is therefore negligible.
\end{proof}



\section{Security of the ClaimChain data structure}\label{app:privacy_proofs}

\subsection{Privacy}
Here we formally describe the privacy properties of \keychains.

\newcommand{\extendchainOracle}{\textsf{EC}}
\newcommand{\adduserOracle}{\textsf{AU}}

\begin{algorithm}[t]
\begin{algorithmic}
\LineComment{Add a new user}
\Procedure{$\adduserOracle$}{\id}
    \Let{$(\sksig^\id, \pksig^\id)$}{$\sig.\keygen(1^\lambda)$}
    \Let{$(\skdh^\id, \pkdh^\id)$}{$\mydh.\keygen(1^\lambda)$}
    \Let{$(\skvrf^\id, \pkvrf^\id)$}{$\vrf.\keygen(1^\lambda)$}
    \Let{$\textsf{keys}^\id$}{$(\sk_*^\id, \pk_*^\id)$}
    
    \Let{$(\sksig'^\id, ...)$}{$\textsf{keys}^\id$} \Comment{{\footnotesize Separately record the signing key}}
\EndProcedure
\\
\LineComment{Extend the chain of an existing user}
\Procedure{$\extendchainOracle$}{$\id, \textsf{data}, \textsf{claims}, \acm, \textsf{store}$}
    \If {user $\id\xspace$ does not exist} \Return $\bot$ \EndIf
    \Let{$\ptr^\id$}{\textsc{ExtendChain}(}
    \State \hspace{5em} $\textsf{data}, \textsf{claims}, \acm, \textsf{keys}^\id \cup \sksig'^\id, \ptr^\id, \textsf{store})$
    \State \Return $\ptr^\id$
\EndProcedure
\end{algorithmic}
\caption{Add user and extend chain oracles}
\label{alg:build_block_oracle}
\end{algorithm}

\descr{Claim privacy.} The adversary cannot learn anything about the content claims for which it does not have the corresponding capabilities.

We formalize this in Experiment~\ref{exp:claimPriv} using an indistinguishability game. The game models that the adversary cannot distinguish between a claim containing one of two equal-length messages of its choice. The experiment is executed by a challenger that plays a game with the adversary $\adv$.

The game starts with creating a user that represents an honest reader, and another user that represents the challenger. We then provide the adversary with an oracle access that allows it to create users and request them to extend their chains with adversary-supplied claims and access control sets (see Algorithm~\ref{alg:build_block_oracle}). Moreover, the adversary is allowed to modify \textsf{store}.

Eventually, the adversary outputs two claims $(l_0, m_0)$ and $(l_1, m_1)$. The challenger flips a random coin $b$, and constructs a challenge block containing claim $(l_b, m_b)$, readable by the honest reader, but not by the adversary. The adversary then has to guess which of the two challenge claims were included in the challenge block. It may make further oracle queries.

Note that this definition implies that the adversary cannot learn anything about the claim neither from the claim encoding itself, not from any of the capabilities. Additionally, the adversary could have access to the claim in the past, but not in the challenge block.

The proof of knowledge $\pi$ in the claim encoding $c$ depends on the claim key $\claimkey$ and other public values, making it difficult to prove directly that the adversary cannot learn anything about the bit $b$. Therefore, in one of the steps we replace this proof $\pi$ with a completely random proof. The following lemma states that we may do so.
\begin{lemma}\label{lem:proof-random}
To any distinguisher that does not know the value $\proofkey \in \{0,1\}^{2\secparameter}$, the proof $\pi$ in \textsc{EncClaim} is indistinguishable from a randomly drawn proof in the random oracle model for $H_q$.
\end{lemma}
\begin{proof}
Without loss of generality, we focus on a simpler proof with only a single conjunct, writing $m$ for $l \parallel \nonce$:
\begin{equation*}
    \pi \gets \textsf{\textsf{SPK}}\{(\skvrf) : \pkvrf = g^{\skvrf} \land h = \vrf.\eval(\skvrf, m) \}(\proofkey).
\end{equation*}
Which abbreviates the following proof:
\begin{equation*}
    \pi \gets \textsf{\textsf{SPK}}\{ (\skvrf) : \pkvrf = g^{\skvrf} \land\\
    h = H_\group(m)^{\skvrf}\}(\proofkey).
\end{equation*}
\newcommand{\skrand}{r_{\textsf{sk}}}
\newcommand{\skresp}{s_{\textsf{sk}}}
To construct this proof, pick a randomizer $\skrand \gets \Zq$, and compute
\begin{align*}
    R_{\textsf{pk}} &= g^{\skrand} \\
    R_h &= H_\group(m)^{\skrand} \\
    c &= H_{\grouporder}(g \parallel H_\group(m) \parallel \pkvrf \parallel h \parallel R_{\pk} \parallel R_h \parallel \proofkey) \\
    \skresp &= \skrand + c \cdot \skvrf.
\end{align*}
The proof is then given by $(c, \skresp).$ To verify the proof, compute
\begin{align*}
    R_{\textsf{pk}}' &= g^{\skresp} \pkvrf^{-c} \\
    R_{h}' &= H_\group(m)^{\skresp} h^{-c},
\end{align*}
and verify that $c$ equals $H_q(g \parallel H_\group(m) \parallel \pkvrf \parallel h \parallel R_{\pk}' \parallel R'_h \parallel \proofkey)$.

Suppose that the adversary does not know $\proofkey$. To randomly generate the proof, draw $(c', \skresp')\sample~\Zq^2$ at random. Since the adversary does not know $\proofkey$ it can never query the random oracle $H_q$ with the correct value for $\proofkey$, therefore it cannot distinguish the fake proof $(c', \skresp')$ from a real proof $(c, \skresp).$
\end{proof}

\begin{theorem}[Claim privacy]\label{thm:claim_privacy}
For any probabilistic polynomial time adversary $\adv$ it holds that $Pr[b = 1] \leq \tfrac{1}{2} + \negl(\lambda),$ where $b \in \{0, 1\}$ is the result of \textsc{ClaimPriv} game (Experiment~\ref{exp:claimPriv}) run with $\adv$. 

\begin{experiment}[t]
\setstretch{1.2}
\begin{algorithmic}
    \ExpTitle{ClaimPriv}
    \LineSep{Setup}
    \State {\Call{$\adduserOracle$}{`reader'}} \Comment{{\footnotesize Initialize reader's chain}}
    \State {\Call{$\adduserOracle$}{`challenger'}} \Comment{{\footnotesize Initialize challenger's chain}}
    \LineSep{Content to include in the challenge block}
    \State $(l_0, m_0), (l_1, m_1), \textsf{data}, \textsf{claims}), \acm, \textsf{store}) \gets $
    \State \hspace{6.5em} $\gets \adv^{\extendchainOracle(\cdot),  \adduserOracle(\cdot)}\hspace{-.1em}\left(\pkdh^{\text{`reader'}}\right )$
    \If {$l_0$ or $l_1$ in $\acm$ or $|m_0| \neq |m_1|$}
        \Return 0
    \EndIf
    \LineSep{Challenge block}
    \Sample{$b$}{$\{0, 1\}$}
    \Let{$\textsf{claims}'$}{$\textsf{claims} \cup \{(l_b, m_b)\}$}
    \Let{$\acm'$}{$\acm \cup \{(\pkdh^{\text{`reader'}}, l_b)\}$} \Comment{{\footnotesize Give the reader the access to $l_b$}}
    \Let{$\ptr_C$}{\Call{$\extendchainOracle$}{`challenger', \textsf{data}, $\textsf{claims}'$, $\acm'$, \textsf{store}}}
    \LineSep{Response}
    \Let{$\hat{b}$}{$\adv^{\extendchainOracle(\cdot),\adduserOracle(\cdot)}\left( \ptr_C \right)$}
    \State \Return $\hat{b} = b$
\end{algorithmic}
\caption{Claim privacy}
\label{exp:claimPriv}
\end{experiment}
\end{theorem}

\begin{proof}
We construct a sequence of games and show that $\adv$ can distinguish between them with negligible probability, starting with $G_0 = \textsf{ClaimPriv}^\adv(\lambda)$.

First, we show that the adversary cannot extract any information about $b$ from the capability entry for $l_b$ because of security of the Diffie-Hellman key exchange and the encryption scheme.

Recall from the \Call{EncCap}{} (Figure~\ref{fig:low_level_algorithms}) and \Call{ExtendChain}{} procedures (Figure~\ref{fig:extend_and_query}) that the corresponding capability lookup key $i_\capab$ and the encryption key $k_\capab$ are given by:
\begin{align*}
    i_\capab &= H_3(s \parallel l_b \parallel \nonce)\\
    k_\capab &= H_4(s \parallel l_b \parallel \nonce),
\end{align*}
where $s$ is the shared DH secret.

\begin{enumerate}[label=$G_\arabic*$]
\item  In this game we substitute the shared Diffie-Hellman secret $s$ with the random string $\alpha \sample \{0, 1\}^{\lambda}$ in all capabilities for reader $\pkdh^{\textrm{`reader'}}$ in all blocks on the challenger's chain. In particular, we set:
\begin{align*}
    i_\capab &= H_3(\alpha \parallel l_b \parallel \nonce)\\
    k_\capab &= H_4(\alpha \parallel l_b \parallel \nonce),
\end{align*}
\item In this game, we substitute the capability key $k_\capab$ with a random string $\beta \sample \{0, 1\}^{2\lambda}$. The capability becomes:
    \[\capab = \enc(\beta, h \parallel \claimkey \parallel \proofkey).\]
\item In this game, we substitute the lookup index $i_\capab$ with a random string $\gamma \sample \{0, 1\}^{2\lambda}$ as well.
\item In this game, we substitute the plaintext $h \parallel \claimkey \parallel \proofkey$ with a random string $\gamma$ of the same length:
    \[\capab = \enc(\beta, \gamma).\]
\end{enumerate}

The games $G_0$ and $G_1$ are indistinguishable by the decisional Diffie-Hellman assumption. Games $G_1$ and $G_2$ are indistinguishable by the pseudorandomness of the hash function $H_4$. The indistinguishability of $G_2$ and $G_3$ follows from the pseudorandomness of $H_3$. Since the encryption key $\beta$ is random, distinguishing between $G_3$ and $G_4$ can be trivially reduced the \textsf{IND-CPA} security for the encryption scheme. Therefore, games $G_3$ and $G_4$ are indistinguishable as well.

The adversary is not allowed to give access to labels $l_0, l_1$ to any user (honest or not). \wlnote{In fact, it would be slightly more accurate to allow it to give access to other honest users, but that makes this proof even longer.} As a result, no other capability entries depend on the challenge bit $b$.

Since $G_4$ replaces the real plaintext with a random plaintext, the adversary also does not learn anything about $\claimkey$ and $\proofkey$.

Now we show that the adversary cannot extract information about $b$ neither from the claim encoding, nor from the claim lookup key. We use the \textsf{IND-CPA} security of the encryption scheme and pseudorandomness of the VRF scheme.

Recall from the \Call{EncClaim}{} (Figure~\ref{fig:low_level_algorithms}) and \Call{ExtendChain}{} procedures (Figure~\ref{fig:extend_and_query}) that here the encoded claim $c$ is given by:
\[c = \enc(\claimkey, \pi \parallel m_b) \parallel \claimcommit.\]

\newcommand{\zkAuthenticator}{H(g^\claimkey) \parallel \elgamal.\enc(g^\claimkey, g^{H(m)})}

\begin{enumerate}[resume, label=$G_{\arabic*}$]
\item In this game, we replace the non-interactive zero-knowledge proof $\pi$ with a uniformly random proof $\pi'$ that does not depend on any of the secret values, nor on any of the public values.
\item In this game, we replace the commitment $\claimcommit$ by a random commitment $\claimcommit_R \gets \group$.
\end{enumerate}
Games $G_4$ and $G_5$ are indistinguishable because of Lemma~\ref{lem:proof-random}. Since $\claimcommitproof$ no longer depends on the randomness $r$, the commitment $\claimcommit$ is perfectly hiding. Therefore, games $G_5$ and $G_6$ are indistinguishable as well.

Next, we change the claim encryption key $\claimkey$ to a random key. Note that because of the changes made in $G_4$, the adversary does not learn anything about $\claimkey$ from the capability $\capab$.

\begin{enumerate}[resume, label=$G_{\arabic*}$]
\item In this game, we generate a random encryption key $\delta$ and use it to replace $\claimkey$:
\begin{equation*}
    c = \enc(\delta, \pi' \parallel m_b) \parallel \claimcommit.
\end{equation*}
\item In this game, we replace the plaintext $\pi' \parallel m_b$ with a random message $\mu$ of the same length:
\begin{equation*}
    c = \enc(\delta, \mu) \parallel \claimcommit_R.
\end{equation*}
\end{enumerate}
Games $G_6$ and $G_7$ are indistinguishable since the adversary learns nothing about $\claimkey$ because of earlier transformations. Games $G_7$ and $G_8$ are indistinguishable because of the CPA security of the encryption scheme.

The final dependency on the bit $b$ is in the claim lookup key $i = H_1(h_b)$, see \Call{EncClaim}{} (Figure~\ref{fig:low_level_algorithms}). We remove this final reference.
\begin{enumerate}[resume, label=$G_{\arabic*}$]
\item In this game, we substitute $h_b$ in $i$ with a random value $q' \sample \mathbb{G}$: \[i = H_1(q')\]
\end{enumerate}
The changes in games $G_4$ and $G_5$ ensure that the adversary does not learn anything about $h_b$ directly. Also, indirectly the adversary cannot learn about $h_b$. 
The adversary can learn \emph{other} VRF values by adding claims and giving itself access to them. However, the pseudorandomness property of the VRF ensures that even if the adversary makes many VRF queries, the remaining values remain pseudorandom. Hence, the adversary cannot distinguish $G_8$ from $G_9$.

In game $G_9$ none of the values depend on the challenge bit $b$, hence, the adversary cannot have advantage better than random guessing.
\end{proof}

\descr{Capability-reader unlinkability.} The adversary should not be able to determine who has been given access to a claim, i.e., for which honest user a capability has been created. We model this using the indistinguishability game in Experiment~\ref{exp:capReaderUnlink}. The adversary can create users (using the $\adduserOracle$ oracle) and extend their chains (using the $\extendchainOracle$ oracle). It then outputs the public keys $\pkdh^0$ and $\pkdh^1$ of two honest users it created using the $\adduserOracle$ and a description of a claim with label $l$ on which it wants to be challenged. The challenger picks one of the honest users at random, and adds a capability to $l$ for that user. The adversary must decide which user has been given the capability.

\begin{theorem}\label{thm:cap_reader_unlink}
For any polynomially-bounded $\adv$ it holds that $Pr[b = 1] \leq  \frac{1}{2} + \negl(\lambda),$ where $b \in \{0, 1\}$ is the result of \textsf{CapReaderUnlink} game (Experiment~\ref{exp:capReaderUnlink}).

\begin{experiment}[t]
\setstretch{1.2}
\begin{algorithmic}
    \ExpTitle{CapReaderUnlink}
    \LineSep{Setup}
    \State {\Call{$\adduserOracle$}{`challenger'}} \Comment{{\footnotesize Initialize challenger's chain}}
    \LineSep{Content to include in the challenge block}
    \Let{$\pkdh^0$, $\pkdh^1$, $l$, $m$, \textsf{data}, $\textsf{claims}$, $\acm$, \textsf{store}}{$\adv^{\extendchainOracle(\cdot), \adduserOracle(\cdot)}()$}
    \If {$\pkdh^0$ or $\pkdh^1$ not a honest user}
        \Return 0
    \EndIf
    \LineSep{Challenge block}
    \Sample{$b$}{$\{0, 1\}$}
    \Let{$\textsf{claims}'$}{$\textsf{claims} \cup \{(l, m)\}$}
    \Let{$\acm'$}{$\acm \cup \{(\pkdh^b, l)\}$}
    \Let{$\ptr_C$}{\Call{$\extendchainOracle$}{`challenger', $\textsf{claims}'$, $\acm'$, \textsf{store}}}
    \LineSep{Response}
    \Let{$\hat{b}$}{$\adv^{\extendchainOracle(\cdot),\adduserOracle(\cdot)}\left( \ptr_C \right)$}
    \State \Return $\hat{b} = b$
\end{algorithmic}
\caption{Capability-reader unlinkability}
\label{exp:capReaderUnlink}
\end{experiment}
\end{theorem}

\begin{proof}
We show that the adversary cannot extract any information about $b$ from the capability entry for $l$. The adversary may have given other readers access to label $l$, but the corresponding capabilities are independent of the bit $b$, so we ignore them. We focus instead on the capability for reader $\pkdh^{b}$. Recall from the \Call{EncCap}{} (Figure~\ref{fig:low_level_algorithms}) and \Call{ExtendChain}{} procedures (Figure~\ref{fig:extend_and_query}) that the corresponding capability lookup key $i_\capab$ and the encryption key $k_\capab$ are given by:
\begin{align*}
    i_\capab &= H_3(s \parallel l_b \parallel \nonce)\\
    k_\capab &= H_4(s \parallel l_b \parallel \nonce),
\end{align*}
where $s$ is the DH secret between the chain owner and the reader $\pkdh^{b}$. We apply the sequence of games $G_0, \ldots, G_4$ in the proof of Theorem~\ref{thm:claim_privacy}. The indistinguishability of the games proves that the adversary does not learn anything about the bit $b$. Therefore, we have capability-reader unlinkability.
\end{proof}

\subsection{Non-equivocation}\label{app:equiv_proofs}

\descr{Intra-block non-equivocation.} Within a given block, a \keychain owner cannot include two different claims having the same label to different readers.

We model this in Experiment~\ref{exp:block_non_eq}. The adversary's task is to produce a block pointed to by $\ptr$ and a label $l$ such that the two readers $\pkdh$ and $\pkdh'$ derive different claims $m$ and $m'$.

\begin{theorem}[Intra-block non-equivocation]\label{thm:block_non_eq}
For any polynomially-bounded $\adv$ it holds that $Pr[b = 1] \leq \negl(\lambda),$ where $b \in \{0, 1\}$ is the result of \textsf{BlockNonEq} game (Experiment~\ref{exp:block_non_eq}).

\begin{experiment}[t]
\setstretch{1.2}
\begin{algorithmic}
    \ExpTitle{BlockNonEq}
    \Let{$\skdh, \pkdh$}{$\mydh.\keygen(1^\lambda)$}
    \Let{$\skdh', \pkdh'$}{$\mydh.\keygen(1^\lambda)$}
    \Let{$l, \ptr, \textsf{store}, \textsf{store}'$}{$\adv(\pkdh, \pkdh')$}
    \Let{$m$}{\Call{GetClaim}{$\skdh, l, \ptr, \textsf{store}$}}
    \Let{$m'$}{\Call{GetClaim}{$\skdh', l, \ptr, \textsf{store}'$}}
    \State \Return $m \neq m' \land m \not= \bot \land m' \not= \bot$
\end{algorithmic}
\caption{Intra-block non-equivocation}\label{exp:block_non_eq}
\end{experiment}
\end{theorem}
\begin{proof}

We first prove that both $\textsf{store}$ and $\textsf{store}'$ must contain the same block $B$. Suppose not, i.e., \textsf{store} contains block $B$ whereas $\textsf{store}'$ contains a different block $B'$ that both hash to the same head $\ptr$. Then the adversary breaks the collision resistance of $H$. Since $H$ is a cryptographic hash function, this happens with negligible probability.

The remainder of this proof is also by contradiction. Assume adversary $\adv$ wins Experiment~\ref{exp:block_non_eq}. We use the uniqueness of the VRF, first for the claim key $\claimkey$, then for the lookup key $h$, to derive a contradiction, i.e., that $m = m'$. 
Both readers $\pkdh$ and $\pkdh'$ first compute the capability lookup key (step 3), see \Call{GetClaim}{} procedure (Figure~\ref{fig:extend_and_query}), retrieve the capability (step 4) and decode it (step 5). Capabilities are per reader, and therefore different. We continue the proof from step 5.

Let $i$ and $i'$ be the claim lookup keys derived in step 5 of the $\Call{GetClaim}$ call by respectively the first and second user. We first consider the case where $i = i'$. By the unique resolution property of the tree (see Experiment~\ref{exp:unique_resolution}), we know that in step 6 both $\Call{GetClaim}$ calls must then derive the same claim encoding $c$ with overwhelming probability.

Since the adversary wins, the derived messages $m$ and $m'$ are different and not $\bot$, therefore the calls to $\Call{DecClaim}$ in step 7 returned different messages $m \not= m'$:
\begin{align*}
    m  &\gets \Call{DecClaim}{\pkvrf^O, l, h, \claimkey, \proofkey, c, \nonce} \\
    m' &\gets \Call{DecClaim}{\pkvrf^O, l, h', \claimkey', \proofkey', c, \nonce}.
\end{align*}
Since the encoding $c$ is the same for both $m$ and $m'$, this situation is not possible by the binding property of the commitment scheme. Indeed, the users verify proofs $\pi$, respectively $\pi'$, in step 6, which verify the commitment $\textsf{com}$.

We now consider the case where the readers derive different lookup keys $i$ and $i'$ in step 5. Since $i \not= i'$ and by the collision resistance of $H_1$, we have that the corresponding VRF values $h$ and $h'$ must be different as well. However, by uniqueness of the VRF, this cannot happen. More precisely, both users successfully verify the proofs $\pi$, respectively $\pi'$, in step 6, which prove that $h = \vrf.\eval($ $\vrftokenExpr) = \claimkey'$, respectively $h' = \vrf.\eval(\vrftokenExpr)$, and therefore $h = h'$, contradicting the assumption that $i \not= i'$.
\end{proof}

\descr{Detectable inter-block equivocation.}\label{thm:inter_block_non_eq}
The game in Experiment~\ref{exp:deteq} models that a claim owner cannot make a non-consistent reference, yet produce a proof of consistency that validates using 
\textsc{CheckConsistency}() (see Figure~\ref{fig:consistency-proof}). More precisely, the adversary outputs valid blocks on two chains: the blocks $\{O_i\}_1^n$ on its own chain, and the blocks $\{C_i\}_1^t$ on the referenced chain. Moreover, the adversary outputs a label $l$ for the referenced chain, and a valid consistency proof $\pi_{\textsf{consist}}$.

To win, the adversary also outputs a pointer $\ptr$ to one of its own blocks such that the challenger has access to label $l$. The adversary wins if the cross-referenced block $m$ differs from the legitimate cross referenced blocks $\{C_i\}_1^t$.



\begin{theorem}\label{thm:deteq}
For any polynomially-bounded stateful $\adv$ it holds that $Pr[d = \top] = \negl(\lambda),$ where $d \in \{\top, \bot\}$ is the result of \textsc{DetEq} game (Experiment~\ref{exp:deteq}). 
\begin{experiment}[t]
\setstretch{1.2}
\begin{algorithmic}
    \ExpTitle{InterBlockEqDetection}
    \LineSep{Setup}
    \State {\Call{$\adduserOracle$}{`challenger'}} \Comment{{\footnotesize Initialize the challenger's chain}}
    \LineSep{Adversary-supplied blocks and validation of consistency}
    \Let{$\{O_i\}_1^n, \{C_i\}_1^t, \textsf{store}, l, \pi_{\textsf{consist}}$}{$\adv^{\adduserOracle(\cdot),\extendchainOracle(\cdot)}()$}
    \If {$\Call{ValidateBlocks}{\{O_i\}_1^n} = \bot$} \Return 0 \EndIf
    \If {$\Call{ValidateBlocks}{\{C_i\}_1^t} = \bot$} \Return 0 \EndIf
    \If {$\Call{CheckConsistency}{l, \{O_i\}_1^n, \{C_i\}_1^t, \pi_{\textsf{consist}}} = \bot$} 
        \State \Return 0
    \EndIf
    \LineSep{Final read phase}
    \Let{$\ptr, \textsf{store}'$}{$\adv()$}
    \If{$\Call{Get}{\textsf{store}', \ptr} \notin \{O_i\}_1^n$}
        \Return 0
    \EndIf
    \Let{$m$}{\Call{GetClaim}{$\skdh^{\textsf{`challenger'}}, l, \ptr, \textsf{store}'$}}
    \State \Return $m \notin \{C_i\}_1^t$
\end{algorithmic}
\caption{Detectable inter-block equivocation}
\label{exp:deteq}
\end{experiment}
\end{theorem}

\begin{proof}
Suppose the adversary wins the game. Let $i$ be the index such that $\ptr$ corresponds to block $O_i$. Since the adversary wins,
\begin{equation*}
    m = \Call{GetClaim}{\skdh^{\textsf{`challenger'}}, l, \ptr, \textsf{store}'}
\end{equation*}
returned a message $m \notin \{C_i\}_1^t$. Let $h$ be the VRF hash that it computes in step 5, and let $c_i = \bar{c_i} \parallel \claimcommit_i$ be the encoded claim that this algorithm retrieves in step 6. In step 7, the algorithm calls $\Call{DecClaim}$, to verify the proof $\pi$. Since the proof is valid, $\claimcommit_i$ commits to $H_q(m)$ and $h_i$ is the VRF hash of $l \parallel \nonce_i$.

We now show that $\Call{CheckConsistency}$ retrieves the same commitment $\claimcommit_i$ together with a proof that the committed value $x' \in \{H_q(C_i)\}_1^t$, contradicting the binding property of the commitment scheme.

The proof $\pi_{\textsf{consist}}$ contains the VRF hash $h_i'$ of $l \parallel \nonce_i$ and the proof of correctness $\pi_h^{(i)}$. Since the proof verified, $h_i'$ is the VRF hash of $l \parallel \nonce_i$, and therefore $h_i' = h_i$. By the unique resolution property of the tree, $\Call{CheckConsistency}$ therefore derived the same encoded claim $c_i = \bar{c_i} \parallel \claimcommit_i$ as the challenger did by calling $\Call{GetClaim}$. Moreover, the proof $\pi_{\textsf{ref}}^{(i)}$ proves that $\claimcommit_i$ commits to $x'$ such that $x' \in \{H_q(C_i)\}_1^t$.

This contradicts the binding property of the commitment scheme or the soundness of the zero-knowledge proofs.
\end{proof}

\end{document}